\documentclass[journal,twoside,web]{ieeecolor}

\usepackage{amsthm}

\usepackage{generic}
\usepackage{cite} 
\usepackage{amsmath,amssymb,amsfonts} 
\usepackage{algorithmic}
\usepackage{graphicx}
\usepackage{algorithm}
\usepackage{textcomp}
\usepackage{float}

\usepackage{color}
\usepackage{relsize}
\usepackage{hyperref} 
\def\BibTeX{{\rm B\kern-.05em{\sc i\kern-.025em b}\kern-.08em
    T\kern-.1667em\lower.7ex\hbox{E}\kern-.125emX}}
\renewcommand{\eqref}[1]{(\ref{#1})}  

\newtheorem{definition}{Definition}
\newtheorem{remark}{Remark}
\newtheorem{theorem}{Theorem}
\newtheorem{lemma}{Lemma}

\newtheorem{proposition}{Proposition}
\newtheorem{fact}{Fact}
\newtheorem{Corollary}{Corollary}

\newcommand{\calG}{\mathcal{G}}
\newcommand{\calV}{\mathcal{V}}
\newcommand{\vect}[1]{\boldsymbol{#1}}

\newcommand{\Vle}[1]{\mathcal{V}_{\leq {#1}}}
\newcommand{\Vl}[1]{\mathcal{V}_{< {#1}}}
\newcommand{\Vge}[1]{\mathcal{V}_{\geq {#1}}}
\newcommand{\Vg}[1]{\mathcal{V}_{> {#1}}}
\newcommand{\calA}{\mathcal{A}}
\newcommand{\calB}{\mathcal{B}}
\newcommand{\calC}{\mathcal{C}}
\newcommand{\Csocial}{C_{\text{social}}}
\newcommand{\Ccog}{C_{\text{cog}}}
\newcommand{\V}{\mathcal{V}}
\newcommand{\G}{\mathcal{G}}
\newcommand{\E}{\mathrm{Ex}\,}
\newcommand{\bO}{\mathbb{O}}
\newcommand{\I}{\mathcal{I}}

\begin{document}
\title{Pareto-Improvement-Driven Opinion Dynamics Explaining the Emergence of Pluralistic Ignorance}
\author{Yuheng Luo$^{1}$, Chuanzhe Zhang$^{1}$, Qingsong Liu$^{2}$, Hai Zhu$^{3}$ and Wenjun Mei$^{1}$
\thanks{This work was supported in part by the National Natural Science Foundation of China under Grants 72531001 and 72201008, and the Beijing Natural Science Foundation under Grant QY24049. We also acknowledge the support of computing resources provided by the High-performance Computing Platform of Peking University.}
\thanks{An earlier conference version of this manuscript has been accepted in the Proceedings of the 64th \emph{IEEE Conference on Decision and Control}~\cite{luo2025pareto}. Compared to the earlier conference version, this journal version contains substantially more results and deeper interpretations. Firstly, in the journal version, we provide detailed proofs for the almost-sure convergence, as well as the necessary and sufficient condition for almost-sure consensus, which are the two most important theorems in this paper. Secondly, all the theoretical results in Section III.C and Section III.D, on the emergence of truth and initial-seeding, are not included in the conference paper. Thirdly, numerical studies and sociological interpretations are not presented in the conference version.}
\thanks{$^{1}$Y. Luo, C. Zhang, and W. Mei are with the Department of Control and Systems Engineering, Peking University, 100871 Beijing, China. $^{2}$Q. Liu is with the School of Artificial Intelligence and Automation, Wuhan University of Science and Technology, 430081 Wuhan, China. $^{3}$ H. Zhu is with the Intelligent Game and Decision Laboratory, Beijing, China. The corresponding author is Wenjun Mei (mei@pku.edu.cn).}
}

\maketitle

\begin{abstract}

Opinion dynamics has recently been modeled from a game-theoretic perspective, where opinion updates are captured by individuals’ cost functions representing their motivations. Conventional formulations aggregate multiple motivations into a single objective, implicitly assuming that these motivations are interchangeable. This paper challenges that assumption and proposes an opinion dynamics model grounded in a multi-objective game framework. In the proposed model, each individual experiences two distinct costs: social pressure from disagreement with others and cognitive dissonance from deviation from the perceived truth. Opinion updates are modeled as Pareto improvements between these two costs. This framework provides a parsimonious explanation for the emergence of pluralistic ignorance, where individuals may agree on something untrue even though they all know the underlying truth. We analytically characterize the model, derive conditions for the emergence and prevalence of the truth, and propose an initial-seeding strategy that ensures consensus on truth. Numerical simulations are conducted on how network density and clustering affect the expression of truth. Both theoretical and numerical results lead to clear and non-trivial sociological insights. For example, no network structure guarantees almost-sure consensus on truth if no one initially expresses the truth; moderately sparse but well-mixed networks are most conducive to consensus on truth. 

\end{abstract}


\section{introduction}

\subsection{Background and motivation} 
Opinion dynamics studies how social interactions and network structures shape collective opinions~\cite{XH-WH-XZ:11,AD-OA:11}. Existing models can be broadly classified into rule-of-thumb dynamics and game-theoretic formulations, where individuals’ opinion updates are characterized by cost functions representing their underlying motivations. In the presence of multiple motivations, these models typically assume an aggregated single-cost objective. However, insights from economics, psychology, and philosophy suggest that different motivations are not always arbitrarily interchangeable~\cite{TRH:99,henderson1992mental,tetlockPsychologyUnthinkableTaboo2000,sunstein1993incommensurability}. This raises an open question: how would collective opinion dynamics change when such an aggregation assumption is relaxed?

In this paper, we propose an opinion dynamics model based on multi-objective games. The model incorporates two cost functions: social pressure arising from disagreement with others’ opinions and cognitive dissonance~\cite{harmonjones2019introduction} caused by deviation from a perceived truth. Opinion updates are formulated as Pareto improvements~\cite{PV:14} between these objectives. We characterize the model’s equilibria, convergence properties, and conditions for convergence to the underlying truth. Moreover, our analysis provides a parsimonious explanation for pluralistic ignorance, i.e., the phenomenon where individuals collectively endorse a false opinion despite knowing the underlying truth~\cite{DGT:82,prentice1993pluralistic,westphal2005pluralistic,o1975pluralistic,miller1987pluralistic}. Numerical studies reveal how social network structures affect collective learning.

\subsection{Literature review}

As one of the earliest formulations of opinion dynamics, the French--DeGroot model assumes that individuals update real-valued opinions through weighted averaging~\cite{FJ-JR:56}, leading to global consensus under mild connectivity conditions~\cite{P-AV-TR:17}. However, this prediction is often inconsistent with the persistent disagreement observed in real social networks~\cite{abelson1964mathematical}. To address this limitation, various extensions have been proposed, including the F-J model with persistent attachment to initial opinions~\cite{F-N-J:90}, the H-K model with bounded confidence~\cite{H-RK-U:02}, and asynchronous variants based on randomized or gossip updates~\cite{boyd2006randomized,parsegov2016novel,deffuant2000mixing}. While asynchronous models generally preserve the macroscopic behaviors of their synchronous counterparts, they may exhibit different asymptotic properties, such as consensus without aperiodicity or ergodic fluctuations induced by stubborn agents~\cite{boyd2006randomized,parsegov2016novel}. Despite these differences, all these models belong to the \emph{rule-of-thumb} class, as they prescribe opinion update rules directly without explicitly modeling the underlying motivations.

Over the past decade, researchers have increasingly adopted a game-theoretic perspective on opinion dynamics~\cite{GP-LJ-WF:14,BD-KJ-OS:15,ES-R-BT:15,NN-MU:16,MW-BF-CG-H-JM-DF:22,MW-H-JM-CG-BF-DF:24}. Groeber et al.~\cite{GP-LJ-WF:14} established a unified micro-foundation based on cognitive dissonance minimization, encompassing the French--DeGroot model and its variants. Mei et al.~\cite{MW-BF-CG-H-JM-DF:22,MW-H-JM-CG-BF-DF:24} further showed that the quadratic social-pressure cost adopted in DeGroot-like models leads to unrealistic implications and derived a weighted-median mechanism from an absolute-value formulation.

Despite these advances, existing game-theoretic models all assume that individuals optimize a weighted sum of multiple costs. This raises a fundamental question: are different psychological motivations always arbitrarily exchangeable? A substantial body of work across economics, psychology, and philosophy suggests otherwise. The theory of incommensurability~\cite{sunstein1993incommensurability} argues that distinct motivations cannot always be quantified on a common scale. Mental accounting theory~\cite{TRH:99,henderson1992mental} suggests that people evaluate outcomes in separate mental ``accounts'' with limited substitutability across categories. The psychology of taboo trade-offs~\cite{tetlockPsychologyUnthinkableTaboo2000} further suggests that certain motivations are inherently resistant to compensation by others.

These perspectives suggest that treating different motivations as non-exchangeable may offer a complementary view for understanding collective opinion phenomena, such as pluralistic ignorance, where individuals privately reject a prevailing norm while publicly conforming to it~\cite{DGT:82,bjerring2014rationality}. A classic example is the campus alcohol consumption study, where students privately opposed binge drinking but publicly conformed due to misperceived peer norms~\cite{prentice1993pluralistic}. Existing models explain pluralistic ignorance either through complex coupling between private beliefs and public expressions~\cite{HC-WT:14,TT-CCG:19,YM-ZL-M:21,yeInfluenceNetworkModel2019} or through binary honest-or-conformist choices~\cite{jadbabaieInferenceOpinionDynamics2023,centola2005emperor,tangStochasticOpinionDynamics2025}. These studies provide valuable insights while leaving open the possibility of understanding pluralistic ignorance from the perspective of non-exchangeable individual motivations.

\subsection{Statement of contributions}
The main contributions of this paper are as follows. 

We propose a novel opinion dynamics framework based on Pareto improvement and multi-objective games. Individuals are characterized by two cost functions representing social pressure and cognitive dissonance, and update their opinions only through Pareto improvements. To our knowledge, this is the first opinion dynamics framework that explicitly models opinion evolution through multi-objective games. The proposed framework investigates collective opinion formation under multiple non-exchangeable motivations and provides a mechanistic explanation for the emergence of pluralistic ignorance.

We then conduct a rigorous theoretical analysis. We characterize the equilibrium set, prove almost-sure finite-time convergence, and derive a necessary and sufficient condition for almost-sure consensus. We establish graph-theoretic conditions for both the prevalence of truth, i.e., consensus on the truth when it is initially expressed, and the emergence of truth, where truthful opinions arise spontaneously. Building on these results, we provide a necessary and sufficient condition for the initial-seeding strategy to guarantee consensus on truth under arbitrary influence networks. The analysis highlights the central role of strictly cohesive sets in determining collective behavior and shows that no network structure ensures almost-sure truth emergence if no one initially expresses it (analogous to the absence of the naïve child in \emph{The Emperor’s New Clothes}). 

Finally, extensive numerical simulations performed on Erdős–Rényi and Watts–Strogatz small-world networks uncover nontrivial relationships between network topology and the expression of truth: High network clustering fosters opinion fragmentation; high edge density tends to prevent consensus on truth, while moderately sparse and well-mixed networks promote truthful consensus.

Technically, compared with weighted-averaging models, the strong nonlinearity and stochasticity of our model prevent direct application of conventional matrix-based analysis methods. We address this challenge by transforming randomness into control inputs, which allows convergence analysis to be reformulated as the problem of constructing suitable update sequences. Moreover, due to the absolute-value cost functions, connectivity alone is insufficient to characterize the system behavior, requiring the identification of more refined network structures. Compared with weighted-median opinion dynamics~\cite{MW-BF-CG-H-JM-DF:22}, our model further incorporates an exogenous perceived truth. This additional objective destroys the contractive property and requires the construction of fundamentally different update sequences for convergence analysis.

\subsection{Organization}
The rest of this paper is organized as follows. Section II introduces basic notions, as well as the setup of the Pareto-Improvement-Driven (PID) opinion dynamics. Section III presents the theoretical analysis and proofs of the main results. Section IV presents the simulation results on different complex networks. Section V is the conclusion.

\section{notations and model setup}
Consider a group of $n$ individuals discussing a certain issue. Their interpersonal influences are characterized by a row-stochastic influence matrix $W = (w_{ij})_{n \times n}$. Denote by $\mathcal{G} (W)$ the directed and weighted graph associated with the adjacency matrix $W$, referred to as the influence network, where each node is an individual. In this paper, we use the terms ``network" and ``graph", ``node" and ``individual", and ``link" and ``edge" interchangeably. Let $\mathcal{V}=\{1,2,\dots,n\}$ be the index set of the nodes, and $w_{ij}>0$ means that there is an edge from node $i$ to node $j$, i.e., individual $i$ is influenced by individual $j$. Denote by 
$\mathcal{N}_i=\{j\in \mathcal{V}|w_{ij}\neq 0\}$ the set of node $i$'s out-neighbors, which includes node $i$ itself if $w_{ii} \neq 0$.

In this paper, we assume that opinions take values from an integer set $\bO=\{x\in\mathbb{Z}\mid m\le x\le M\}$, where $|\bO|\ge2$. This formulation naturally applies to ordered opinions or decisions, including Likert-scale attitude measurements, quantitative decision tasks (e.g., determining the number of solar panels for residential installation), and choices with ordered options (e.g., political party preferences when parties are ordered along an ideological spectrum). It is not intended for unordered categorical choices, such as selecting among different electric vehicle brands.

We further assume that there exists an opinion $\theta\in\bO$ that every individual privately recognizes as true. The uniqueness of the truth is consistent with existing theoretical and modeling studies on pluralistic ignorance~\cite{centola2005emperor,bjerring2014rationality}. One may then ask why the population does not simply converge to $\theta$. The key reason is that others' recognition of the truth is not common knowledge: individuals observe only others' expressed opinions rather than their private beliefs, and may therefore misperceive the prevailing social norm, which discourages them from revealing their true opinions initially.

Let $x\in \bO^n$ be the vector of everyone's expressed opinion. Denote by $x_i$ the $i$-th component of $x$, i.e., node $i$'s expressed opinion. Since we only model the dynamics of expressed opinions, the term ``expressed'' is usually omitted for simplicity. Define the sublevel set of $z$ as
\begin{align*}
   \Vle{z}(x)=\{ j\in \mathcal{V}\,|\,x_j\le z\} 
\end{align*}
and the superlevel set $\Vge{z}(x)\!=\!\{ j\in \mathcal{V}\,|\,x_j\!\ge\! z\}$. The notations $\Vl{z}(x)$, $\Vg{z}(x)$, and $\calV_{=z}(x)$ are defined analogously. 

Each individual aims to minimize two types of costs. The first one is the social pressure one experiences from expressing an opinion different from others', called the \emph{social cost}. For any individual $i$, such a cost is given by
\begin{align}\label{eq:social}
    \Csocial^i(x_i;x,W)=\sum_{j=1}^n w_{ij}|x_i-x_j|.
\end{align}
In the rest of this paper, without causing any confusion, we omit $W$ from the notation $\Csocial^i(x_i;x,W)$. Note that the widely-studied DeGroot-like models~\cite{D-MH:74} adopt a quadratic form of the above cost, i.e., $\sum_{j=1}^n w_{ij}(x_i-x_j)^2$. However, as pointed out by Mei et al.~\cite{MW-BF-CG-H-JM-DF:22}, quadratic cost functions lead to the non-negligibly unrealistic implication that opinion attractiveness increases with opinion distance. Therefore, in this paper, we adopt the absolute-value form as in~\cite{MW-BF-CG-H-JM-DF:22}. The second type of cost is the cognitive dissonance caused by deviating from one's perceived truth, called the \emph{cognitive cost}. For any individual $i$, such a cost is simply given by
\begin{align}\label{eq:cog}
    \Ccog^i(x_i;\theta)  =|x_i-\theta|.
\end{align}

Motivated by the perspectives reviewed in Section~I.B, we investigate the opposite extreme of the conventional aggregated-cost formulation. While the latter assumes that different motivations are fully exchangeable through a weighted sum, we assume that they are completely non-exchangeable and therefore cannot be aggregated into a single objective. Under this assumption, opinion updates are governed by Pareto improvement~\cite{PV:14}, i.e., improving at least one objective without worsening the others. Individuals update their opinions asynchronously. Whenever an individual is activated, they choose an opinion from the finite set $\bO$ that constitutes a Pareto improvement with respect to the two cost functions, as specified below.
\begin{definition}[PID opinion dynamics]\label{def:PID-op-dyn}
Consider a group of $n$ nodes embedded in an influence network $\calG(W)$. Assume that individuals take opinions from a finite set $\bO$ and $\theta\in \bO$ is the truth perceived by every individual. Denote by a vector $x(t)\in \bO^n$ the individuals' opinions at time $t$. The Pareto-Improvement-Driven (PID) opinion dynamics is defined as the following discrete-time stochastic process: At each time $t+1$, a node $i$ is uniformly randomly picked from $\mathcal V$ and randomly chooses an opinion as its new opinion from the set
\begin{align*}
    P_i(x(t))=\Big{\{} z \in \bO\,\Big|\, & \Csocial^i(z;x(t)) \leq \Csocial^i(x_i(t);x(t)),\\
    & \Ccog^i(z;\theta) \leq \Ccog^i(x_i(t);\theta) \Big{\}},
\end{align*}
i.e., the set of all the Pareto improvements between the costs $\Csocial^i$ and $\Ccog^i$. An update sequence is called ``legal'' if it satisfies the PID opinion dynamics.
\end{definition}
Note that our definition of Pareto improvement does not require a strict improvement of one of the two costs, which is slightly different than the conventional setup~\cite{PV:14}. 

\begin{remark}
     Compared with classical opinion dynamics models, our model is formulated using two objective functions that are not aggregated into a single objective. A closely related model is the weighted-median prejudice opinion dynamics~\cite{zhangConvergenceAnalysisWeightedMedian2025a}, which aggregates an absolute-value social cost and individual prejudice (another form of reference opinion) into a single objective. In that model, if all individuals share the same prejudice, asymptotic consensus is guaranteed (see Theorem~4.1 in~\cite{zhangConvergenceAnalysisWeightedMedian2025a}). By contrast, as shown later in Theorem~\ref{th:equi}, PID opinion dynamics may converge either to a non-consensus equilibrium or to a consensus on a non-truth opinion. This contrast highlights that consensus formation becomes more challenging when multiple motivations are treated as non-exchangeable. Moreover, as a model of pluralistic ignorance, the proposed framework has a considerably more parsimonious structure while allowing a richer spectrum of opinions to characterize gradual opinion formation.
\end{remark}

To provide an intuitive illustration of the PID opinion dynamics, we simulate the model on a lattice graph. As shown in Fig.~\ref{fig:lat}, the system exhibits rich dynamical behavior. Depending on the initial condition and the sequence of node updates, the population may reach a consensus on the truth $\theta$, a consensus on a false opinion distinct from $\theta$, or even settle on a non-consensus equilibrium.  In the next section, we conduct a comprehensive theoretical analysis of the PID opinion dynamics.

\begin{figure*}[t]  
    \centering
    \includegraphics[width=1.0\linewidth]{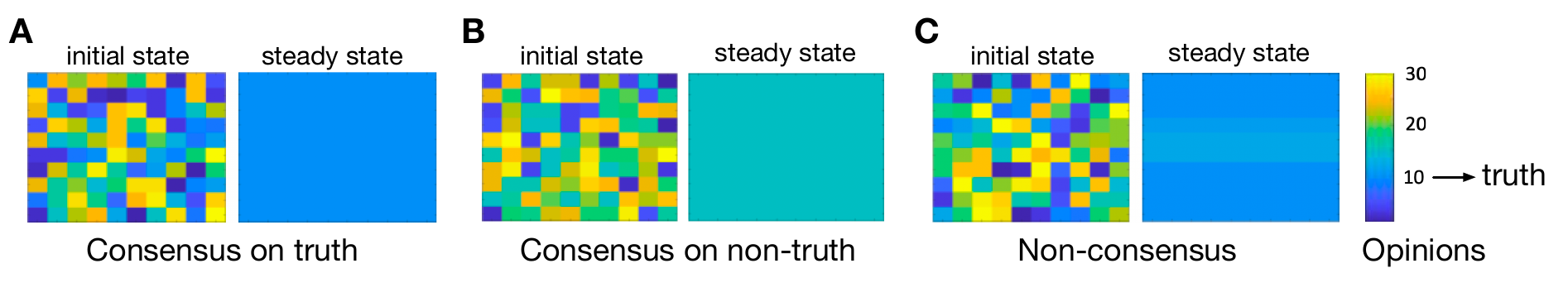}  
    \caption{Simulations of the PID opinion dynamics on a lattice. The system comprises a 10×10 lattice network (100 nodes arranged uniformly in a grid), where each node connects bidirectionally via directed edges to its von Neumann neighbors (up, down, left, right). Each node’s outgoing edges are assigned equal weights normalized to sum to 1. We set $\{1,2,...,30\}$ to be the set of available opinions and set $\theta=10$. The opinion each color represents and the color of the truth are presented on the right. The left picture of each panel is the initial condition and the right is the equilibrium achieved. Panels~\textbf{A},~\textbf{B}, and~\textbf{C} illustrate consensus on truth, consensus on non-truth, and non-consensus equilibrium respectively. }
    \label{fig:lat}
\end{figure*} 

\section{theoretical analysis of the PID opinion dynamics}

\subsection{Basic properties: equilibria and convergence}

 We first introduce a fine structure of influence networks that plays an important role in our analysis.

\begin{definition}[(Strictly) cohesive set]\label{def:scs}
Given a graph $\calG(W)$ with the node set $\mathcal{V}$, a cohesive set $\mathcal{M} \subseteq \mathcal{V}$ is a subset of nodes that satisfies $\sum _{j \in \mathcal{M}}w_{ij}\geq \frac{1}{2}$ for any $i\in \mathcal{M}$. If `$\geq$' is replaced by `$>$', then $\mathcal{M}$ is referred to as a strictly cohesive set. An empty set is considered strictly cohesive. 
\end{definition}

By Definition~\ref{def:scs}, if $w_{ii}>\frac{1}{2}$, then $i$ itself forms a strictly cohesive set. Definition~\ref{def:scs} naturally leads to an operation called the \emph{cohesive expansion}, formalized below. This concept will be employed in subsequent proofs.
\begin{definition}[Cohesive expansion]\label{def:CE}
Given an influence network $\cal{G}(W)$ with the node set $\cal{V}$ and a subset of nodes $\mathcal{M}\subseteq \mathcal{V}$, a subset of $\mathcal{V}$ is a cohesive expansion of $\mathcal{M}$ (written simply as $\E(\mathcal{M})$) if it can be constructed via the following iteration algorithm:
\begin{enumerate}[label=\arabic*)]
\item Let $\mathcal{M}_0=\mathcal{M}$;
\item For $k=0,1,2,...$, if there exists $i\in \mathcal{V} \setminus \mathcal{M}_k$ such that $\sum_{j\in \mathcal{M}_k}w_{ij}\geq\frac{1}{2}$, then let $\mathcal{M}_{k+1}=\mathcal{M}_k\cup \{i\}$;
\item Terminate the iteration at step $k$ as long as there does not exist any $i\in \mathcal{V}\setminus \mathcal{M}_k$ satisfying $\sum_{j\in \mathcal{M}_k}w_{ij}\geq \frac{1}{2}$. let $\E(\mathcal{M})=\mathcal{M}_k$.
\end{enumerate}
\end{definition}

The following fact states that the cohesive expansion of any node set is unique and thus well-defined. Its proof follows the same reasoning as Lemma~3 in~\cite{MW-H-JM-CG-BF-DF:24}. Note that our definition differs slightly from that in~\cite{MW-H-JM-CG-BF-DF:24}, as we use $\sum_{j\in \mathcal{M}_k} w_{ij} \ge 1/2$ instead of $> 1/2$ as the criterion for adding node $i$ to $\mathcal{M}_k$. The proof of uniqueness is not affected by this minor difference and is thus omitted.

\begin{fact}[Uniqueness of cohesive expansion]\label{fac:uq-ce}
Given an influence network $\calG(W)$ with the node set $\cal{V}$, for any $\mathcal{M}\subseteq \mathcal{V}$, the cohesive expansion of $\mathcal{M}$ is unique, i.e., independent of the order of node additions.
\end{fact}

The following theorem characterizes the equilibria of the PID opinion dynamics in terms of strictly cohesive sets. 

\begin{theorem}[Set of equilibria]\label{th:equi}
Consider an influence network $\mathcal{G}(W)$ with n nodes, and let $\bO=\{x \in \mathbb{Z} \mid m \le x \le M\}$. An opinion vector $x^* \in \bO^n$ is an equilibrium of the PID opinion dynamics given by Definition~\ref{def:PID-op-dyn}, i.e., $x^*$ satisfies $P_i(x^*)=\{x_i^*\}$ for any $i \in \mathcal{V}$, if and only if one of the two statements hold:
\begin{enumerate}[label=\arabic*)] 
  \item $x^*$ is a consensus state, i.e., $x_1^*=x_2^*=\cdots=x_n^*$.
  \item For any $ z\in \bO$ and $z<\theta$, $\mathcal{V}_{\leq z}(x^*)$ is a strictly cohesive set in $\calG(W)$; for any $z\in \bO$ and $z>\theta$, $\mathcal{V}_{\geq z}(x^*)$ is a strictly cohesive set.
    \end{enumerate}
\end{theorem}
\begin{proof}
We first prove the ``if'' part. If Statement 1) holds, we have $\Csocial^i(x^*_i;x^*)=0$ for any $ i \in \mathcal{V}$. As a result, any deviation from $x_i^*$ will increase node $i$'s social cost. In this case, $P_i(x^*) = \{x_i^*\}$ for any $i\in \V$ and $x^*$ is thus an equilibrium. If Statement 2) holds and, without loss of generality, suppose $\calV_{<\theta}(x^*)$ is non-empty. We first prove that, for any given $i\in \Vl{\theta}(x^*)$, individual $i$ will not update their opinion to any value larger than $x_i^*$. According to Statement~2), $\Vle{x_i^*}(x^*)$ is a strictly cohesive set, i.e., 
\begin{align}\label{eq:proof-thm-equi-strict-cohesive}
    \sum_{j:\, x_j^*\le x_i^*} w_{ij} - \sum_{j:\, x_j^*>x_i^*} w_{ij}>0.
\end{align}
For any $z\in \mathbb{R}$ and $z>x_i^*$, 
\begin{align*}
    \Csocial^i(z;x^*) & = \sum_{j:\, x_j^*\le z} w_{ij}(z-x_j^*) + \sum_{j:\, x_j^*>z} w_{ij}(x_j^*-z)\\
    & = \Big( \sum_{j:\, x_j^*\le z} w_{ij} - \sum_{j:\, x_j^*>z} w_{ij} \Big) z \\
    & \quad + \sum_{j:\, x_j^*>z} w_{ij} x_j^* - \sum_{j:\, x_j^*\le z} w_{ij} x_j^*.
\end{align*}
According to~\eqref{eq:proof-thm-equi-strict-cohesive}, we have
\begin{align*}
    \sum_{j:\, x_j^*\le z} w_{ij} - \sum_{j:\, x_j^*>z} w_{ij} \ge \sum_{j:\, x_j^*\le x_i^*} w_{ij} - \sum_{j:\, x_j^*>x_i^*} w_{ij}>0. 
\end{align*}
Therefore, for any $z\in (x_i^*,+\infty)$, $\Csocial^i(z;x^*)$ is a continuous piece-wise affine function with strictly positive slopes, which implies that $\Csocial^i(z;x^*)$ strictly increases with $z$ in the interval $(x_i^*,+\infty)$. Therefore, for any $z\in \bO$ and $z>x_i^*$,  we have $\Csocial^i(z;x^*)>\Csocial^i(x_i^*;x^*)$ and thereby $z\notin P_i(x^*)$. Now consider any $z\in \bO$ and $z<x_i^*$. Since $z<x_i^*<\theta$, we have
\begin{align*}
    \Ccog^i(z;\theta) = |z-\theta| > |x_i^*-\theta| = \Ccog^i(x_i^*;\theta).
\end{align*}
That is, updating from $x_i^*$ to any $z\in \bO$ with $z<x_i^*$ will increase $i$'s cognitive cost $\Ccog^i$. Therefore, we have
\begin{align*}
    P_i(x^*) = \{x_i^*\}.
\end{align*}
Following the same line of argument, one could also prove that, for any $i\in \Vg{\theta}(x^*)$, $P_i(x^*)=\{x_i^*\}$. In addition, for any $i$ such that $x_i^*=\theta$, any deviation from $\theta$ increases the cognitive cost $\Ccog^i$ and thereby $P_i(x^*)=\{x_i^*\}$. This concludes the proof for the ``if'' part.

Now we proceed to the proof for the ``only if'' part. Suppose that an opinion vector $x^*\in \bO^n$ satisfies neither Statement~1) nor Statement~2). Without loss of generality, assume that there exists $i\in \V$ such that $x_i^*<\theta$ and $\Vle{x_i^*}(x^*)$ is not strictly cohesive. In this sense, there exists $j\in \Vle{x_i^*}(x^*)$ such that
\begin{align*}
    \sum_{k:\, x_k^*>x_i^*} w_{jk} \ge \frac{1}{2}.
\end{align*}
Let $\mu = \min \{y\in \bO\,|\, y>x_j^*\}$. Since $x_i^*<\theta$, such $\mu$ must exist and $\mu\le \theta$. For any $z\in (x_j^*, \mu)$,
\begin{align*}
    \Csocial^j(z;x^*) & = \sum_{k:\, x_k^*\le x_j^*} w_{jk}(z-x_k^*) + \sum_{k:\, x_k^*>x_j^*} w_{jk}(x_k^*-z)\\ 
    & = \Big( \sum_{k:\, x_k^*\le x_j^*} w_{jk} - \sum_{k:\, x_k^*>x_j^*} w_{jk} \Big)z \\
    &\quad + \sum_{k:\, x_k^*>x_j^*} w_{jk}x_k^* - \sum_{k:\, x_k^*\le x_j^*} w_{jk}x_k^*,
\end{align*}
which is an affine function of $z$ with a non-positive slope. Due to the continuity of $\Csocial^j(z;x^*)$ in $z$, we have $\Csocial^j(\mu;x^*)\le \Csocial^j(x_j^*;x^*)$. Moreover, since $x_j^*<\mu\le \theta$, we have $\Ccog^j(\mu;\theta)<\Ccog^j(x_j^*;\theta)$. Therefore, $\mu\in P_j(x^*)$ and thereby $x^*$ is not an equilibrium. This concludes the proof for the ``only if'' part.
\end{proof}

\begin{remark}
Intuitively, the sublevel set $\mathcal{V}_{\le z}(x^*)$ represents a group of individuals whose opinions do not exceed $z$. When this group forms a strictly cohesive set, the internal social influence among its members outweighs the influence exerted by the rest of the network. In this sense, it behaves as a self-sustaining ``echo chamber'', where no individual experiences sufficient social pressure to raise their opinion above $z$.
\end{remark}

\begin{remark}
    In weighted-averaging models, connectivity is the key structural property for consensus, because the underlying quadratic cost implies increasingly strong attraction as opinion distance grows~\cite{MW-BF-CG-H-JM-DF:22}. This excessive attraction leaves little room for finer local structures. In contrast, under absolute-value cost functions (including our model), attraction no longer increases with opinion distance, allowing (strictly) cohesive sets to resist consensus. Although a lack of strong connectivity implies multiple (strictly) cohesive sets, the converse does not hold because cohesive sets also depend on edge weights.
\end{remark}

As stated in Theorem~\ref{th:equi}, the equilibria of the PID opinion dynamics fall into three categories, see examples in Fig.~\ref{fig:lat}. The first corresponds to consensus on truth, i.e., $x^* = \theta\vect{1}_n$. The second category is dissensus, where there exist $i, j \in \calV$ such that $x_i^* \neq x_j^*$. The third and most intriguing category is false consensus, i.e., consensus at some $z \neq \theta$. In this case, although every individual privately recognizes the true value $\theta$, the group collectively “agrees” on an untrue opinion due to social pressure and perceived norms. Equilibria of this kind capture the psychological phenomenon of pluralistic ignorance within our model. As explained by our model, this phenomenon arises naturally from the non-exchangeability of individuals’ multiple motivations. Notably, compared with existing theories, the PID opinion dynamics focus solely on the evolution of expressed opinions driven by two simple motivational forces, making the model both minimal and interpretable. 

Now we proceed to the convergence analysis of the PID opinion dynamics. We first present a useful lemma. Its proof is based on the properties of absorbing states of Markov chains, see Lemma~7 in~\cite{MW-H-JM-CG-BF-DF:24}.
\begin{lemma}[Transforming randomness to sequence design] \label{lem:tf-rd-sd}
Consider the PID opinion dynamics given by Definition 1. If, starting from any $x\in \bO^n$, there exists a legal and finite update sequence $(i_1,z_1),\cdots,(i_{T},z_{T})$, along which the opinion trajectory reaches an equilibrium at time step $T$, then the PID opinion dynamics almost surely converges to an equilibrium in finite time, for any initial condition $x(0) \in \bO^n$. (Here $(i_t,z_t) $ means that node $i_t$ updates to opinion $z_t$ at time $t$)
\end{lemma}

\begin{remark}
The termination time $T$ in Lemma~\ref{lem:tf-rd-sd} depends on the initial state $x(0)$. Since the set of possible initial conditions is finite, the termination times over all initial states admit a finite upper bound. Therefore, the dependence of $T$ on $x(0)$ does not affect the validity of this lemma.
\end{remark}

The following theorem establishes the almost-sure convergence of the PID opinion dynamics.
\begin{theorem}[Almost-sure finite-time convergence]\label{th:as-ft-cv}  
Consider the PID opinion dynamics given by Definition~\ref{def:PID-op-dyn}, with $\bO=\{x \in \mathbb{Z} \mid m \le x \le M\}$. For any initial state $x(0) \in \bO^n$, $x(t)$ almost surely reaches an equilibrium in finite time.
\end{theorem}

\begin{proof}    
We first point out a simple fact: once an individual $i$ expresses the opinion $\theta$, their opinion will always remain unchanged because any deviation from the truth $\theta$ would increase their cognitive dissonance $\Ccog^i$.

According to Lemma~\ref{lem:tf-rd-sd}, the PID opinion dynamics almost surely converge in finite time if, for any given initial condition $x(0)\in \bO^n$, we can construct a legal update sequence, along which the solution $x(t)$ reaches an equilibrium in finite time. Obviously, if $x(0)$ is a consensus state, then $x(0)$ is already an equilibrium. From now on we suppose $x(0)$ is not a consensus state. 

Before presenting the formal steps, we sketch the update sequence. We treat the two sides of $\theta$ separately and never allow updates to cross $\theta$. On the side $z\le \theta$:
\begin{enumerate}
    \item We start from the closest occupied opinion  below $\theta$ whose sublevel set is not strictly cohesive (denoted as $\psi^0$ in the proof).
    \item  We then pull inward: repeatedly update any node $i$ that is farther from $\theta$ on the same side to the current threshold opinion (e.g., to $\psi^0$) if the sum of the weights on the edges from node $i$ to nodes with opinions below $\psi^0$ is less than $0.5$, until $\calV_{<\psi^0}(\cdot)$ becomes empty or strictly cohesive. This is essentially the construction process of $\E(\calV_{\ge \psi^0}(x(0)))$.
    \item Then we move to the next closest occupied opinion below $\theta$ (denoted $\psi^1$), and iterate the same inward “coalescing” procedure. 
\end{enumerate}
The above procedures will generate a finite sequence $\psi^0,\psi^1,\ldots$ that monotonically sweeps inward toward $\theta$, rendering every sublevel set $V_{<z}(\cdot)$ (for $z\le \theta$) strictly cohesive (or empty). We then mirror the same inward sweep on the side $z\ge \theta$, starting from the closest occupied opinion above $\theta$ whose superlevel set is not strictly cohesive (denoted $\phi_0$), and iterating analogously. After these two operations, all sub/superlevel sets around $\theta$ are strictly cohesive (or empty), so no further Pareto-improving update is possible and an equilibrium is reached. The detailed argument is specified as follows.

\emph{Step 1 (starting point of iteration):} Let 
\begin{align*}
    \psi^0 & = \max\{z\in \bO\,|\, z\le \theta,\, \Vl{z}(x(0))\text{ is non-empty and} \\
    & \qquad \qquad\qquad\,\,\,\,\,\text{not strictly cohesive}\}\\
    \varphi^0 & = \min\{z\in \bO\,|\, z\ge \theta,\, \Vg{z}(x(0))\text{ is non-empty and} \\
    & \qquad \qquad\qquad\,\,\,\,\,\text{not strictly cohesive}\}.
\end{align*}
If $\psi^0$ does not exist, then either $\Vl{\theta}(x(0))$ is empty or $\Vl{z}(x(0))$ is a strictly cohesive set for any $z\le\theta$ and $z\in \bO$. In the latter case, according to equation~\eqref{eq:proof-thm-equi-strict-cohesive} and the argument below it in the proof of Theorem~\ref{th:equi}, $P_i(x(0))=\{x_i(0)\}$ for any $i\in \Vl{\theta}(x(0))$. The same argument also applies to the case when $\varphi^0$ does not exist. Therefore, if neither $\psi^0$ nor $\varphi^0$ exists, then $x(0)$ is already an equilibrium.

\emph{Step 2 (iterations of nodes on one side of $\theta$):} Without loss of generality, suppose $\psi^0$ exists. For any $t\ge 0$, if there exists $i\in \Vl{\psi^0}(x(0))$ such that 
\begin{align*}
    \sum_{j:\, x_j(t)\ge \psi^0} w_{ij}\ge \frac{1}{2} \ge \sum_{j:\, x_j(t)<\psi^0} w_{ij},
\end{align*}
then one could easily check that 
\begin{align*}
    \Csocial^i(\psi^0;x(t)) & \le \Csocial^i(x_i(t);x(t)),\text{ and}\\
    \Ccog^i(\psi^0;\theta) & < \Ccog^i(x_i(t);\theta).
\end{align*}
Let individual $i$ update their opinion to $\psi^0$, which is a legal update under the Pareto-improvement rule.
Since $\Vl{\psi^0}(x(0))$ is a finite set, the above procedure must terminate at some finite time $t_1$. Note that, from $t=0$ to $t_1$, the node update order is the same as a feasible node addition sequence for the cohesive expansion $\E(\calV_{\ge \psi^0}(x(0)))$. 

At time $t_1$, $\Vl{\psi^0}(x(t_1))$ is either empty or non-empty and strictly cohesive. This is because that, if $\Vl{\psi^0}(x(t_1))$ is non-empty, there exists no node $i$ with $\sum_{j:\, x_j(t_1)<\psi^0} w_{ij}\le \frac{1}{2}$, which implies that $\Vl{\psi^0}(x(t_1))$ is strictly cohesive. Suppose $\Vl{\psi^0}(x(t_1))$ is non-empty, let 
\begin{align*}
    \psi^1 = \max_{i\in \Vl{\psi^0}(x(t_1))} x_i(t_1).
\end{align*}
Then $\Vl{\psi^1}(x(t_1))$ must be in one of the two cases below:

\emph{Case 1:} $\Vl{\psi^1}(x(t_1))$ is empty, which means that all the individuals in $\Vl{\psi^0}(x(t_1))$ hold the opinion $\psi^1$ at time $t_1$ and $\{i\in \calV\,|\, x_i(t_1)=\psi^1\}=\Vl{\psi^0}(x(t_1))$ forms a strictly cohesive set. Therefore, individuals in $\Vl{\psi^0}(x(t_1))$ will never change their opinions from now on. In this case, terminate Step 2 and proceed to Step 3;

\emph{Case 2:} $\Vl{\psi^1}(x(t_1))$ is non-empty. In this case, if $\Vl{\psi^1}(x(t_1))$ is not strictly cohesive, we have $\E(\Vge{\psi^1}(x(t_1)))\neq \emptyset$. Applying the same operation as performed on $\psi^0$, we update all nodes in $\E(\Vge{\psi^1}(x(t_1)))\setminus \Vge{\psi^1}(x(t_1))$ to $\psi^1$ following the node addition order as the cohesive expansion $\E(\Vge{\psi^1}(x(t_1)))$. This procedure ends at time step $t_2$, when either $\Vl{\psi^1}(x(t_2))$ is empty, which is Case 1, or $\Vl{\psi^1}(x(t_2))$ is strictly cohesive. If $\Vl{\psi^1}(x(t_1))$ is already strictly cohesive, we will not let nodes in $\calV_{<\psi^1}(x(t_1))$ update to $\psi^1$ even though there may exist some nodes that can do so. Let $t_2=t_1$ for notational uniformity on this occasion. 

At time step $t_2$, for any $\psi \in \bO$ and $\psi^1\le\psi\le \theta$, $\calV_{<\psi}(x(t_2))$ is strictly cohesive. Therefore, for any individual holding opinion $\psi$, updating their opinion to any $z>\psi$ will strictly increase the social cost $\Csocial^i$. In the meanwhile, updating their opinion to any $z<\psi$ will strictly decrease the cognitive cost $\Ccog^i$. Therefore, 
\begin{align*}
    P_i(x(t_2))=\{x_i(t_2)\}
\end{align*}
for any $i$ with $x_i(t_2) \in \bO \cap [\psi^1,\theta]$.

Now we further update the opinions of individuals in $\calV_{<\psi^1}(x(t_2))$ as follows. For any $k\ge 1 $, if $\calV_{<\psi^k}(x(t_{k+1}))$ is non-empty, let \begin{align*}
    \psi^{k+1}=\max_{i:x_i(t_{k+1})<\psi^k} x_i(t).
\end{align*}
Repeat the above procedures described in Case 1 and Case 2. Since $\calV_{<\psi^0}(x(0))$ is a finite set, such procedures must terminate at some time $t_K$ and $\calV_{<\psi ^{K-1}}(x(t_K))$ is empty. As a consequence, $\calV_{<z}(x(t_K))$ is strictly cohesive for any $z\le \theta$ and \begin{align*}
    P_i(x(t_K))=\{x_i(t_K)\}, \,\forall i \in \calV_{\le \theta}(x(t_K)).
\end{align*}

\emph{Step 3 (the same iterations on the other side of $\theta$):}
Now we deal with the nodes holding opinions no less than $\theta$. Since Step 2 only involves updates for nodes with opinions less than $\theta$, and no node updates across the ``mid-point'' $\theta$, we have $\calV_{> \theta}(x(t_K))=\calV_{> \theta} (x(0))$. If $\varphi ^0$ does not exist, then, according to the discussion in Step 1, $P_i(x(t_K))=\{x(t_K)\}$ for any $i\in \calV_{\ge \theta}(x(t_K))$ and $x(t_K)$ is thereby already an equilibrium. If $\varphi^0$ exists, update the opinions of individuals in $\calV_{\ge \theta}(x(t_K))$ via a similar procedure as in Step 2. Such updates will also terminate at some finite time $t_{K'}$, when $P_i (x(t_{K'}))=\{x_i(t_{K'})\} $ for any $i\in \calV_{> \theta}(x(t_{K'}))$.

\emph{Step 4 (check for whether an equilibrium is reached):}
For any $t\in \{t_{K+1},...,t_{K'}\}$, only individuals with $x(t)>\theta$ update their opinions, and their opinions are not updated to any value less than $\theta$. Therefore, the property ``$\calV_{<z}(x(t))$ is strictly cohesive for any $z\le \theta$'', established at the end of Step 2, still holds at $t=t_{K'}$. Therefore, for any $i\in \calV$, $P_i(x(t_{K'}))=\{x_i(t_{K'})\}$. That is, along the update sequence constructed above, $x(t)$ reaches an equilibrium at some finite time $t_{K'}$. This concludes the proof. 
\end{proof}  
Theorem~\ref{th:as-ft-cv} indicates that the PID opinion dynamics almost surely reach an equilibrium in finite time, independent of network structures and initial states. This property justifies the PID opinion dynamics as a well-posed model. 

\subsection{Consensus and prevalence of the truth}
In this subsection, we examine the network conditions under which the PID opinion dynamics almost surely reach consensus. Building on this analysis, we analyze conditions for the prevalence of truth.

The following theorem establishes a graph-theoretic necessary and sufficient condition for almost-sure convergence to consensus from any initial state.

\begin{theorem}[Necessary and sufficient condition for consensus]\label{th:cd-cs}  Consider the PID opinion dynamics given by Definition~\ref{def:PID-op-dyn}, with $\bO=\{x \in \mathbb{Z} \mid m \le x \le M\}$. The following two statements are equivalent:
\begin{enumerate}[label=\arabic*)]
\item For any $x(0) \in \bO^n$, the solution $x(t)$ almost surely converges to consensus in finite time.
\item The only strictly cohesive set in $\calV$ is $\calV$ itself.
 \end{enumerate}
 \end{theorem}
\begin{proof} "1) $\Rightarrow$ 2)": We prove this part by showing that when the network contains more than one strictly cohesive set, there exist a path along which the system converges to a dissensus equilibrium. Suppose there exists a strictly cohesive set $ \calV_1 \subset \mathcal{V}$. We now construct an initial condition which constitutes a non-consensus equilibrium. For any $ i \in \calV_1$, let $ x_i(0)=z_1, z_1 \in \bO \setminus{\theta} .$  For any $ i \in \mathcal{V}\setminus \calV_1$ , let $x_i(0)=\theta.$ It follows that $x(0)$ satisfies Statement 2) in Theorem~\ref{th:equi} and is thereby an equilibrium. In addition, if the initial condition is randomly generated from $\bO^n$, then such a $x(0)$ occurs with non-zero probability. Therefore, the almost-sure convergence to consensus does not hold.

"2) $\Rightarrow$ 1)": It suffices to prove that all the equilibrium points are consensus states under Statement 2). Suppose, on the contrary, there exists a non-consensus equilibrium $x^*$. Then there exists $i,j\in \mathcal{V}$ such that $x_i^*\neq x_j^*$. By symmetry and without loss of generality, assume $x_i^*\leq\theta$. We show that both of the following two cases lead to contradiction.
\begin{itemize}
    \item Case 1: If $|x_i^*-\theta|\ge|x_j^*-\theta|$, then $x_i^*\neq \theta$ (Otherwise, $x_j^*=\theta=x_i^*$). Moreover, since $x_i^*\le \theta$ and $x_i^*\neq x_j^*$, we have $x_i^*<x_j^*$. Since $\calV$ is the only strictly cohesive set in $\calV$ and $ \Vle{x_i^*}(x^*)\neq \V$, $\calV_{\leq x^*_i}(x^*)$ is not a strictly cohesive set. By Statement 2) in Theorem~\ref{th:equi}, $x^*$ is not an equilibrium.
    \item Case 2: If $|x_i^*-\theta|<|x_j^*-\theta|$, since $x_i^*\le \theta$, we have either $x_j^*<x_i^*\le \theta$ or $x_i^*<\theta<x_j^*$. In the former case, $\Vle{x_j^*}(x^*)\neq \V$ and is thus not strictly cohesive; In the latter case, $\Vge{x_j^*}(x^*)\neq \V$ and is thus not strictly cohesive. According to Theorem~\ref{th:equi}, in any of these two cases, $x^*$ is not an equilibrium.  
\end{itemize}
The discussion of the two cases above concludes the proof. \end{proof}

Theorem~\ref{th:cd-cs} suggests that consensus can be guaranteed only when the network contains no proper strictly cohesive sets, where internal influence dominates external influence. We next analyze the conditions for the prevalence of truth, i.e., when all individuals eventually express the opinion $\theta$ provided that at least one individual initially does so. The following corollary shows that initial truth-tellers can propagate the truth throughout the network if and only if the network admits no strictly cohesive sets other than the entire network.

\begin{Corollary}[ Necessary and sufficient condition for prevalence of truth]\label{th:cd-cv-th}  
Consider the PID opinion dynamics given by Definition~\ref{def:PID-op-dyn}, with $\bO=\{x \in \mathbb{Z} \mid m \le x \le M\}$. Let $\mathcal{X}_0=\{x(0)\in \bO^n\,|\, \exists i\in \calV \,\mathrm{s.t.,} \, x_i(0)=\theta\}$. The following two statements are equivalent:
\begin{enumerate}[label=\arabic*)]
\item For any $x(0)\in \mathcal{X}_0$, and along any update sequence, the system almost surely converges to consensus on $\theta$ in finite time.
\item The only strictly cohesive set in $\mathcal{V}$ is $\mathcal{V}$.
\end{enumerate}
\end{Corollary}
\begin{proof}``2) $\Rightarrow$ 1)'': By Theorem~\ref{th:cd-cs}, the system almost surely converges to consensus in finite time. Consider any $x(0)\in \mathcal{X}_0$. Since there exists $i\in \mathcal{V}$ such that $x_i(0)=\theta$, then $C^i_{cog}(x_i(0);\theta)=0$. The Pareto-improvement rule guarantees that $C^i_{cog}$ is non-increasing, therefore we have $x_i(t)\equiv \theta$. Therefore, consensus can only be achieved on $\theta$. Incorporating the system's almost-sure finite-time convergence property, we conclude that the system almost-surely converges to consensus on $\theta$. 

``1) $\Rightarrow$ 2)'': We prove the contrapositive. Suppose there exists a nonempty strictly cohesive subset $\calV_1 \neq \calV$. Construct an initial state in which all nodes in $\calV_1$ hold the same opinion $x_0 \in \bO \setminus {\theta}$, while all nodes in $\calV \setminus \calV_1$ hold the opinion $\theta$. It is straightforward to verify, by Theorem~\ref{th:equi}, that this configuration is already an equilibrium of the PID opinion dynamics. Consequently, there exists an initial condition $x(0) \in \mathcal{X}_0$ under which the system fails to reach consensus along any update sequence. Furthermore, if the initial condition is randomly drawn from $\mathcal{X}_0$, such a configuration occurs with nonzero probability. Therefore, almost-sure convergence to consensus on $\theta$ does not hold. This completes the proof.

\end{proof}

Note that the necessary and sufficient condition ``$\V$ is the only strictly cohesive set'' in Theorem~\ref{th:cd-cs} and Corollary~\ref{th:cd-cv-th} differs from the necessary and sufficient condition for almost-sure consensus in the weighted-median opinion dynamics~\cite{MW-H-JM-CG-BF-DF:24}, which is ``$\V$ is the only maximal cohesive set''. In the weighted-median model, individuals' opinion updates are only driven by social pressure. The condition ``$\V$ is the only strictly cohesive set'' admits a natural sufficient condition: “The only cohesive set in $\calV$ is $\calV$ itself.” These two conditions differ only by a marginal threshold and are therefore generically equivalent. To provide a clearer interpretation of the former, we next present an intuitive and equivalent characterization of the latter. Its proof is provided in Appendix~\ref{app:proof-traverse}.

\begin{proposition}\label{prop:traverse}
Consider an influence network $\calG (W)$ with the node set $\mathcal{V}$. The node set $\calV$ itself is the only cohesive set in $\calV$ if and only if the edges in $\calG(W)$ with $w_{ij}>\frac{1}{2}$ form a directed cycle that traverses all the nodes in $\calV$. 
\end{proposition}

\subsection{Emergence of the truth}
In this subsection, we analyze the system’s behavior when no individual initially expresses opinion $\theta$. Specifically, we derive the conditions under which the PID opinion dynamics can achieve consensus on the truth under this constraint. In the rest of this subsection, we set $\theta$ to $0$ for brevity of presentation. If $\theta \neq 0$, one can subtract $\theta$ from all opinions and set $\theta$ to $0$. This transformation does not affect the values of $\Csocial$ and $\Ccog$, and thus leaves the dynamical process invariant.

Several lemmas are introduced as intermediate steps to establish the main theorems in this subsection.
\begin{lemma}[Property of pareto-improvement set]\label{lemma:prop-pis}
Consider the PID opinion dynamics given by Definition~\ref{def:PID-op-dyn} with $\bO=\{x \in \mathbb{Z} \mid m \le x \le M\}$. For any $x\in \bO^n$ and any $i\in \V$, the Pareto-improvement set $P_i(x)$ satisfies that, for any $\alpha,\, \beta\in P_i(x)$, if $\alpha<\beta$, then $[\alpha,\beta]\cap \bO^n \subseteq P_i(x)$. 
\end{lemma}
The proof of the above lemma is provided in Appendix~\ref{app:proof-lemma:prop-pis}. The following definition characterizes a specific type of opinion updates in the PID opinion dynamics. 

\begin{definition}[Crossing update]\label{def:trv}
Consider the PID opinion dynamics given by Definition~\ref{def:PID-op-dyn}. The opinion update at time step $t$ is a crossing update if there exists $i\in \V$ such that $(x_i(t)-\theta)(x_{i}(t-1)-\theta)<0$. Namely, node $i$ updates its opinion across $\theta$ at time $t$.
\end{definition}

\begin{lemma}[Existence of update sequence  without crossing  update]\label{lem:ex-wt-trv}
Consider the PID opinion dynamics given by Definition~\ref{def:PID-op-dyn}. Let $\bO=\{m,...,0,...,M\}$ with $m\leq -1$, $M\ge1$ and $\theta=0$. For any $x(0)\in \bO ^n$, if there exists a finite update sequence along which the opinion trajectory reaches consensus on 0, then there exists a finite update sequence without crossing updates along which the opinion trajectory reaches consensus on 0. 
\end{lemma}
The proof is provided in Appendix~\ref{app:proof-lem:ex-wt-trv}. This lemma reveals an interesting fact: crossing updates are unnecessary for reaching a truthful consensus. Building on this result, we derive the following lemma to prepare for establishing a necessary condition for the existence of an update sequence that leads to consensus on the truth in finite time.

\begin{lemma}[Existence of update sequence to $\pm$1]\label{lem:ex-pm1}
Consider the PID opinion dynamics as in Definition~\ref{def:PID-op-dyn}. Let $\bO=\{m,...,0,...,M\}$ with $m\leq -1$, $M\ge1$ and $\theta=0$. For any $x(0)\in \bO^n$, if there exists a finite-time update sequence without crossing update, along which the opinion trajectory reaches consensus on $0$, then there also exists a finite update sequence without crossing update, along which the system reaches a state (not necessarily an equilibrium) such that 
every node in $\calV_{<0}(x(0))$ expresses the opinion $-1$, while every node in $\calV_{>0}(x(0))$ expresses the opinion $1$.
\end{lemma}

The proof is provided in Appendix~\ref{app:proof-lem:ex-pm1}. This lemma constructs an update sequence along which all nodes whose initial expressed opinions differ from $\theta$ move to $\pm 1$. Since each node’s distance from $0$ must be monotonically non-increasing along the PID opinion dynamics, the scenario in which all nodes' expressed opinions start in $\{-1,0,1\}$ can be regarded as a case where the set of available opinions is restricted to $\{-1,0,1\}$. The proof of Lemma~\ref{lem:ex-pm1} directly yields the following corollary, which will be used in the proof of the main results in this subsection.

\begin{Corollary}\label{coro:extreme theta}
Consider the PID opinion dynamics as in Definition~\ref{def:PID-op-dyn}. Let $\bO=\{0,1,...,M\}$ with $M\ge1$ and $\theta=0$. For any $x(0)\in \bO^n$, if there exists a finite-time update sequence along which the opinion trajectory reaches consensus on $0$, then there also exists a finite update sequence, along which the system reaches a state (not necessarily an equilibrium) such that every node in $\calV_{>0}(x(0))$ expresses the opinion $1$.    
\end{Corollary}

The following lemma states that, if $\bO=\{-1,0,1\}$ and no one initially express the truth $\theta=0$, then convergence to consensus on $\theta$ fails with non-zero probability. Its proof is provided in Appendix~\ref{app:proof-lem:equi-pm10}. Note that the condition $\bO=\{-1,0,1\}$ here is adopted only to simplify the presentation, and can be generalized as $\bO=\{\theta-1,\theta,\theta+1\}$, since the PID opinion dynamics is invariant under constant shifts of opinion values.

\begin{lemma}[Equilibrium when $\bO=\{-1,0,1\}$]\label{lem:equi-pm10}
Consider the PID opinion dynamics given by Definition~\ref{def:PID-op-dyn}. Suppose that $\bO=\{-1,0,1\}$ and $\theta=0$. If  
 $\calV_{=0}(x(0))$ is empty, then there exists a finite update sequence leading the system to either consensus on $z\neq 0$ or a non-consensus equilibrium.
\end{lemma}

The theorem below generalizes the results in Lemma~\ref{lem:equi-pm10} from $\bO=\{-1,0,1\}$ to any arbitrary finite set $\bO$.

\begin{theorem}[Almost-sure emergence of truth cannot be guaranteed]\label{th:nc-as-cv-cs-th}
Consider the PID opinion dynamics given by Definition~\ref{def:PID-op-dyn}. Let $\bO=\{x \in \mathbb{Z} \mid m \le x \le M\}$, where $m$ and $M$ are not both equal to $\theta$=0. If $\calV_{=0}(x(0))$ is empty, then there exists an update sequence along which the system achieves non-truth consensus or non-consensus equilibrium in finite time.
\end{theorem}
\begin{proof}
We distinguish two cases depending on whether $\theta = 0$ is an extreme opinion. Here, saying that $\theta = 0$ is an extreme opinion means that $\bO$ contains only non-negative or only non-positive opinions.
    
We first consider the case where $\theta = 0$ is not extreme, i.e., $\bO = \{m, \ldots, 0, \ldots, M\}$ with $m \leq -1$, and $M \ge 1$. The analysis of this case builds on Theorem~\ref{th:as-ft-cv}, Lemmas~\ref{lem:ex-wt-trv},~\ref{lem:ex-pm1}, and~\ref{lem:equi-pm10} via the following steps:
\begin{enumerate}
    \item By Theorem~\ref{th:as-ft-cv}, there exists a finite update sequence that drives the system to an equilibrium. If the resulting equilibrium is a non-truth consensus or a non-consensus equilibrium, the claim is already established.
    \item Otherwise, suppose there exists an opinion-update sequence along which the system reaches consensus on $0$. Then, by Lemma~\ref{lem:ex-wt-trv}, there exists a finite update sequence with no crossing updates along which the opinion trajectory reaches consensus on $\theta = 0$. Moreover, by Lemma~\ref{lem:ex-pm1} and the assumption that $\calV_{=0}(x(0))$ is empty, there exists a finite update sequence $S$ that moves all nodes to states in $\{-1, 1\}$. 
    \item Execute such an update sequence $S$. Denote by $T$ the termination time of $S$. Since $\Ccog^i(x_i(t))$ is non-increasing along any legal update sequence, each node’s distance to $0$ does not increase. As a result, after time $T$, the available opinion set is confined to $\{-1, 0, 1\}$. 
    \item By Lemma~\ref{lem:equi-pm10}, there then exists a finite update sequence that drives the system from $x(T)$ to either consensus on some $z \neq 0$ or to a non-consensus equilibrium.
\end{enumerate}

Secondly, consider the case where $\theta = 0$ is extreme. Without loss of generality, suppose all elements in $\bO$ are non-negative. We again assume that there exists an update sequence $S$ that drives the system to consensus on $\theta =0$. From Corollary~\ref{coro:extreme theta} we know that there also exists a finite update sequence $S'$ along which every $i \in \calV_{>0}(x(0))$ moves to $1$. Combining this with the fact that $\calV_{>0}(x(0)) = \calV$, we conclude that, along the update sequence $S'$, which will be triggered with non-zero probability, the system reaches consensus on $1$ in finite time. This completes the proof.
\end{proof}

\begin{remark}
Theorem~\ref{th:nc-as-cv-cs-th} can be equivalently interpreted as follows: “$\calV_{=\theta}(x(0))$ is nonempty” is a necessary condition for almost-sure convergence to consensus on $\theta$. This result conveys a clear message that, if no one initially speaks out the truth, the almost-sure emergence of truth is not guaranteed under any influence network structure.
\end{remark}

Next, we establish a sufficient condition on the network structure and the initial state for the existence of an update sequence leading to consensus on $\theta$.

\begin{theorem}[Sufficient Condition for Truth to Possibly Emerge]\label{th:tpe} Consider the PID opinion dynamics given by Definition~\ref{def:PID-op-dyn}. Let $\bO=\{x \in \mathbb{Z} \mid m \le x \le M\}$, where $m$ and $M$ are not both equal to $\theta$=0. For any $x(0)\in \bO^n$, if neither $\calV_{<0}(x(0))$ nor $\calV_{>0}(x(0))$ contains a non-empty strictly cohesive set, then there exists a finite opinion-update sequence along which the system reaches consensus on $\theta=0$.
\end{theorem}

\begin{proof}
For simplicity of notations, let
\begin{align*}
    \calV_0^{-} = \calV_{<0}(x(0))\quad \text{and}\quad \calV_0^+ = \calV_{>0}(x(0)).
\end{align*}
Without loss of generality, suppose $\calV_0^-$ is non-empty. We now construct a finite opinion-update sequence that drives the opinions of all the nodes in $\calV_0^-$ to $\theta=0$. 

We first prove that $\calV_0^-\subset\E\big( \calV_{\ge 0}(x(0)) \big)$. Since $\calV_0^-$ contains no non-empty strictly cohesive set, itself is not strictly cohesive. Therefore, there exists $i_1\in \calV_0^-$ such that 
\begin{align*}
    \sum_{j\in \calV_{\ge 0}(x(0))} w_{i_1 j}\ge \frac{1}{2}
\end{align*}
In turn, since $\calV_0^-\setminus \{i_1\}$ is not strictly cohesive either, there exists $i_2\in \calV_0^-\setminus\{i_1\}$ such that 
\begin{align*}
    \sum_{j\in \calV_{\ge 0}(x(0))\cup\{i_1\}} w_{i_2j}\ge \frac{1}{2}
\end{align*}
Since $\calV_0^-$ contains no non-empty strictly cohesive set, as the above argument continues, we will eventually obtain a node-addition sequence $(i_1,i_2,\dots,i_T)$ such that 
\begin{align*}
    \calV_0^-=\{i_1,\dots, i_T\}\text{ and }\E\big( \calV_{\ge 0}(x(0)) \big)=\calV_{\ge 0}(x(0))\cup \calV_0^-.
\end{align*}
As a result, for any $t\in \{1,\dots, T\}$, 
\begin{align*}
    \sum_{j\in \calV_{\ge 0}(x(0))\cup \{i_1,\dots,i_{t-1}\}}w_{i_t j}\ge \frac{1}{2}.
\end{align*}
According to the Pareto-improvement mechanism, the above inequality in turn implies that
\begin{align*}
    \Csocial^i\big( 0;x(t-1) \big)\le \Csocial^i\big( x_{i_t}(t-1);x(t-1)  \big).
\end{align*}
In addition, $\Ccog^i(0)\le \Ccog^i(x_{i_t}(t-1))$ always holds. Therefore $(i_t,0)$ constitutes a legal opinion update at time $t$. That is, the node-addition sequence $(i_1,\dots, i_T)$ for $\E\big( \calV_{\ge 0}(x(0)) \big)$ in fact induces a legal opinion-update sequence for the PID opinion dynamics:
\begin{align*}
    S=\Big( (i_1,0),\,(i_2,0),\dots,\, (i_T,0) \Big),
\end{align*}
after which we have $x_i(T)=0$ for any $i\in \calV_0^-$. To this point, all the nodes with initial opinions below 0 have updated their opinions to 0.

Now we deal with the nodes in $\calV_0^+$. Since no node in $\calV_0^+$ has ever updated its opinion up to time $T$, we have $\calV_{>0}(x(T))=\calV_0^+$. Moreover, since $\calV_0^+$ does NOT contain any non-empty strictly cohesive set, following the similar treatment above for $\calV_0^-$ and starting from time $T+1$, we can construct a finite and legal opinion-update sequence along which the opinions of all the nodes in $\calV_0^+$ are updated to $0$. As a result, the entire system reaches consensus on $0$ along a finite and legal opinion-update sequence. This concludes the proof. 
\end{proof}

Theorem~\ref{th:tpe} indicates that strictly cohesive sets among individuals deviating from the truth in either direction may obstruct the propagation of truth. When neither of the two groups of non-truth opinions contains a strictly cohesive set, there exists a feasible sequence of opinion updates leading the entire network to express the truth with positive probability.

\subsection{Guiding the system to truth via initial seeding}

In this subsection, we consider the following questions: If we choose a subset of nodes as ``initial seeds'' and let their initial opinions be $\theta$, which nodes should be chosen so that the entire system is guaranteed to reach consensus on truth? How is the choice of initial seeds determined by the influence network structure? The following theorem answers the above questions by providing a necessary and sufficient graph-theoretic condition for a seed-selection strategy to guarantee almost-sure consensus on truth.

\begin{theorem}[Consensus on truth via initial seeding] \label{th:cs-at-th} 
Consider the PID opinion dynamics given by Definition~\ref{def:PID-op-dyn}, with $\bO=\{x \in \mathbb{Z} \mid m \le x \le M\}$. Suppose a subset of nodes, denoted by $\I$, start with the opinion $\theta$. The following two statements are equivalent:
\begin{enumerate}[label=\arabic*)]
\item The system almost surely converges to consensus on $\theta$, regardless of the initial opinions of the nodes in $\V\setminus \I$.
\item Every strictly cohesive set in $\G(W)$ contains at least one node in $\I$.
\end{enumerate}
\end{theorem}

\begin{proof}
We first point out a simple fact that, according to the Pareto-improvement mechanism, once a node's opinion is $\theta$, its opinion will never change along the PID opinion dynamics.

\emph{``2) $\Rightarrow$ 1)'':} Since $\calV$ is a strictly cohesive set, there exists at least one node in $\mathcal{I}$, i.e., one node who will stick to the opinion $\theta$ throughout the dynamics. Therefore, the system can only reach either consensus on $\theta$ or a non-consensus equilibrium. Suppose the system reaches a non-consensus equilibrium $x^*$. Then there exists $j\in\calV$ such that $x^*_j\neq \theta$. Without loss of generality, assume $x^*_j<\theta$. According to statement 2) in Theorem~\ref{th:equi}, $\calV_{\le x^*_j}$ is a strictly cohesive set and thereby contains a node in $\I$, whose opinion is fixed at $\theta$. This contradicts the definition of $\calV_{\le x^*_j}$ and thus conclude the proof for ``2) $\Rightarrow$ 1)''. 
      
\emph{``1) $\Rightarrow$ 2)'':}
We prove its contrapositive. Assume that there exists a strictly cohesive set $\calV_0$ such that no node in $\calV_0$ is in $\mathcal{I}$. We construct an initial condition as follows: All the nodes in $\calV_0$ are initially at $x_0 \in \bO \setminus{\{\theta\}}$, while other nodes are initially at $\theta$. According to Theorem~\ref{th:equi}, such an initial condition is already a non-consensus equilibrium. This concludes the proof for ``1) $\Rightarrow$ 2)''. 
       
\end{proof}
   
Theorem~\ref{th:cs-at-th} highlights the crucial role of strictly cohesive sets in enabling consensus on truth via initial seeding. However, given an arbitrary influence network, identifying all strictly cohesive sets is a computationally demanding task, reflecting the inherent complexity of social systems. The following is a direct corollary of Theorem~\ref{th:cs-at-th}  under a particular graph-theoretic condition.
\begin{Corollary}[Truth Consensus via a globally reachable node]~\label{cor:grn}
    Consider the PID opinion dynamics given by Definition~\ref{def:PID-op-dyn}, with $\bO=\{x \in \mathbb{Z} \mid m \le x \le M\}$. Suppose that in $\calG(W)$, there exists a node globally reachable via paths consisting exclusively of edges with weights greater than $0.5$. If this globally reachable node initially expresses $\theta$, the system almost surely converges to  consensus on $\theta$.
\end{Corollary}
The proof  is provided in Appendix~\ref{app:proof-cor:grn}.

\section{Numerical study}

Since the PID opinion dynamics assumes the existence of an objective truth $\theta$, the extent to which individuals reach consensus on $\theta$ reflects the system’s level of ``collective intelligence''. In this section, we perform extensive numerical simulations to explore the role of influence network structure in promoting collective intelligence. We focus on two basic structural metrics: link density and clustering coefficient. The former captures how well-connected the network is, while the latter indicates the degree of local clustering versus global mixing.

\begin{figure*}[t]  
    \centering
    \includegraphics[width=1\linewidth]{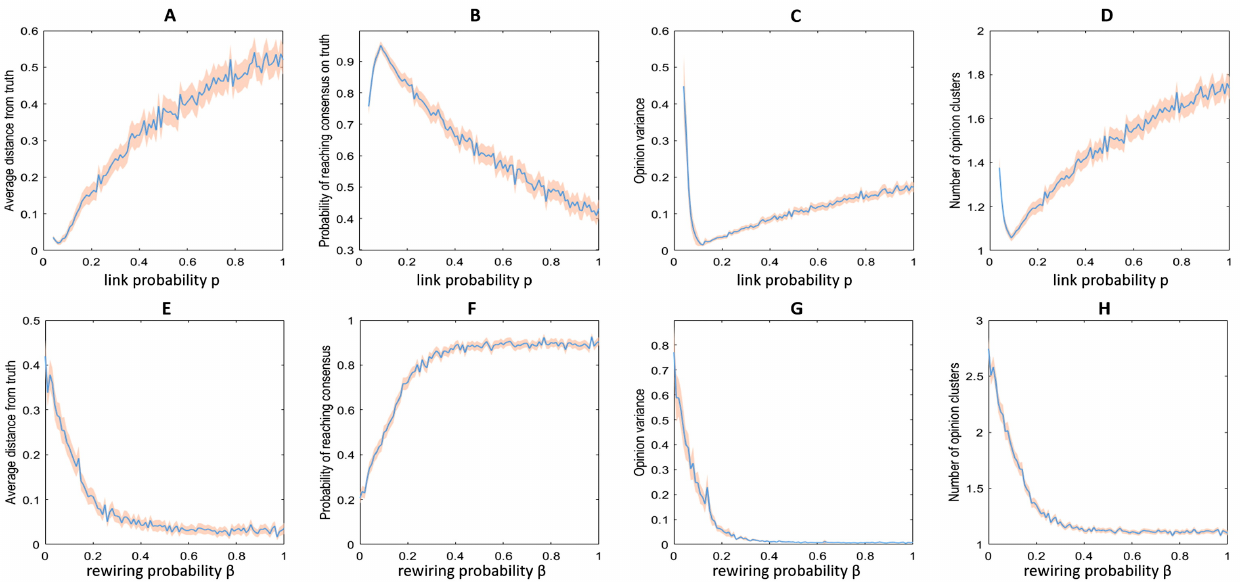}  
    \caption{How link density and clustering coefficient affect the system's ability of reaching consensus on truth. Panel~\textbf{A}-\textbf{D} correspond to the results of simulations on E-R random networks, showing how link probability $p$ affects the average opinion distance to truth, probability of reaching consensus on truth, variance of steady-state opinion, and number of opinion clusters at the steady state, respectively. Panel~\textbf{E}-\textbf{H} correspond to the results of simulations on W-S small-world networks, showing how rewiring probability $\beta$ affects the average opinion distance to truth, probability of reaching consensus on truth, variance of steady-state opinion, and number of opinion clusters at the steady state, respectively. In each plot, the color-shaded regions demarcate the 95\% confidence interval for the value computed based on 1000 independent simulations under the same $p$. It seems that there is no color-shaded region in Panel~\textbf{C} and ~\textbf{D} because the 95\% confidence intervals are too narrow.}
    \label{fig:simulations}
\end{figure*}

\emph{Effects of link density:} Simulations of the PID opinion dynamics on Erd\"{o}s–R\'{e}nyi (E–R) graphs~\cite{PE-AR:59} reveal non-trivial patterns on how link density affects collective intelligence. We simulate our dynamics on E-R networks with 100 nodes. In E-R networks, each directed edge is independently generated with probability $p$, and only strongly connected networks are retained. Edge weights are then randomly sampled from $(0,1]$ and normalized to ensure that each node's outgoing weights sum to 1. In our simulations, we let $p$ take values from the set $\{0.04,0.05,\dots,0.99,1\}$. For each value of $p$, 1000 independent simulations are conducted. For each simulation, an E-R network with link probability $p$ is randomly generated. Let $\bO=\{1,2,\dots,30\}$ and uniformly randomly designate one opinion in $\bO$ as the truth $\theta$. Each node draws its initial opinion uniformly at random from $\mathbb{O}$.

Statistical analysis shows that, when $p$ is sufficiently small, the average distance between opinions and $\theta$ decreases as $p$ increases, reaching a minimum at $p = 0.09$, and then increases (see Fig.\ref{fig:simulations}\textbf{A}). This trend is corroborated by the probability of reaching consensus on $\theta$ (Fig.\ref{fig:simulations}\textbf{B}). The opinion variance and the number of opinion clusters exhibit nearly identical trend curves: both drop rapidly as $p$ rises, and then increase slowly with growing $p$ once $p$ exceeds around 0.09 (Fig.\ref{fig:simulations}\textbf{C} and \textbf{D}).

 When $p<0.09$, sparse networks contain many cohesive subgroups that prevent truth consensus. To investigate why the truth consensus rate decreases with increasing edge density when $p>0.09$, we conduct additional simulations and find that the effect of \(\theta\) on the truth consensus rate differs significantly across edge densities.
 The experimental configuration is specified as follows.
 We fix the edge probability at two values, \(p=0.1\) and \(p=1\), and vary \(\theta\) from 1 to 30 across the full opinion set \(\mathbb{O}\), instead of drawing \(\theta\) uniformly at random. All other experimental settings remain identical to those in the previous experiment. For each \((p,\theta)\) combination, 1000 independent simulations are performed, and we record the truth consensus rate with each $\theta$. As shown in Figure.\ref{fig:theta}\textbf{A}, when \(p=0.1\), the truth‑consensus probabillity exhibits little variation for \(\theta\in[5,25]\). It decays as the distance from \(\theta\) to the median element of \(\mathbb{O}\) (i.e., 15) increases. When \(\theta\) lies at the extreme edges of \(\mathbb{O}\), the truth‑consensus  probabillity is approximately \(0.5\). Outside this interval, the  probabillity drops accordingly. By contrast, as shown in Figure.\ref{fig:theta}\textbf{B}, for \(p=1\), the system’s truth‑consensus rate varies linearly with the distance between \(\theta\) and 15. The truth‑consensus rate approaches zero when \(\theta\) is 1 or 30. This illustrates that the difference in truth‑consensus rate between the two extreme cases \(p=0.1\) and \(p=1\) is strongly dependent on the value of \(\theta\). To explore the underlying cause of this phenomenon, we examine the characteristics of non‑consensus equilibria for networks with \(p=0.1\) and \(p=1\) under \(\theta=2\) ($\bO=\{1,2,..,30\}$).  Representative simulation results are shown in Figure.\ref{fig:network}. For both \(p=0.1\) and \(p=1\), nodes concentrate near \(\theta=2\). When $p = 0.1$, a large fraction of nodes settle at opinion 3, while at $p = 1$, a substantial number of nodes are distributed across both values 3 and 4.

\begin{figure*}[t]
  \centering
  \includegraphics[width=1\linewidth]{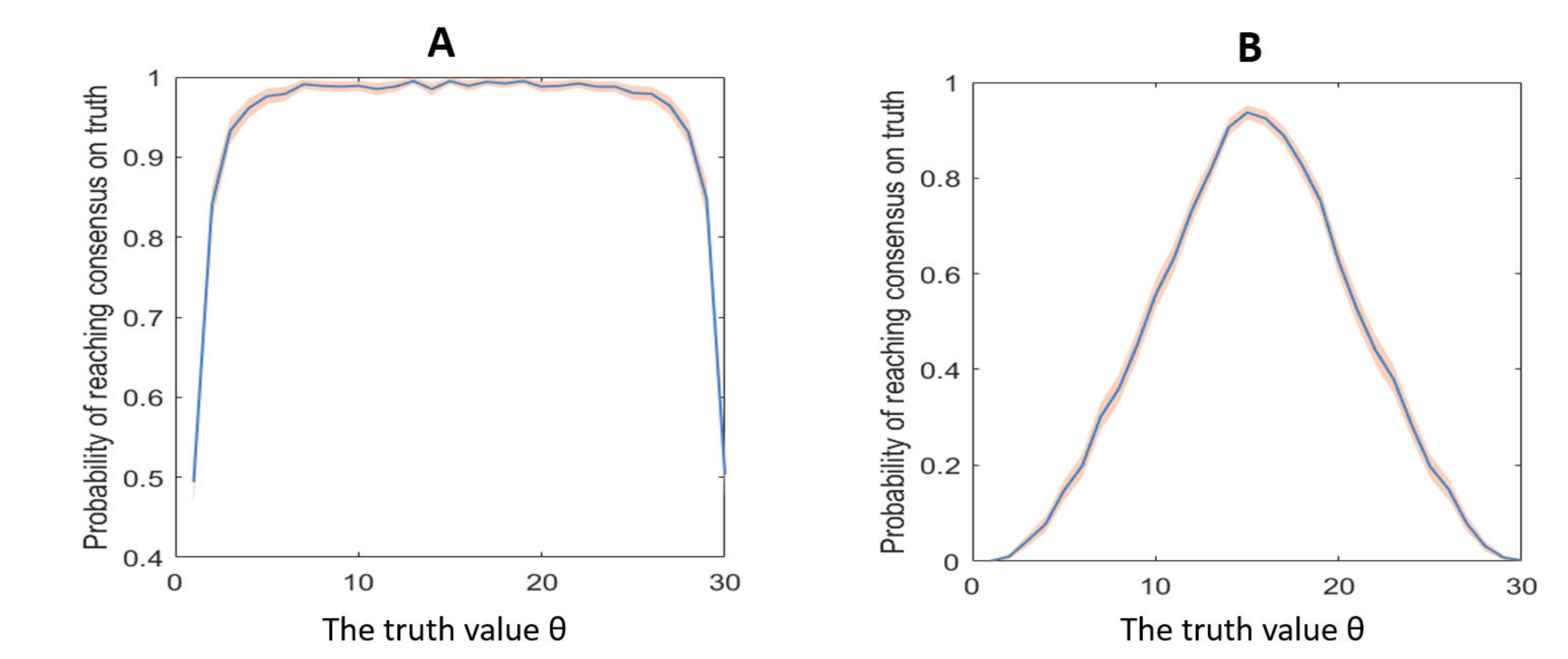}
  \caption{How the value of the truth $\theta$ affect the probability of consensus on the truth. Panel~\textbf{A} corresponds to the results of simulations when $p=0.1$.  Panel~\textbf{B} corresponds to the results of simulations when $p=1$. In each plot, the color-shaded regions demarcate the $95\%$ confidence interval for the value computed based on 1000 independent simulations under the same $p$.}
  \label{fig:theta}
\end{figure*}

\begin{figure*}[t]
  \centering
  \includegraphics[width=1\linewidth]{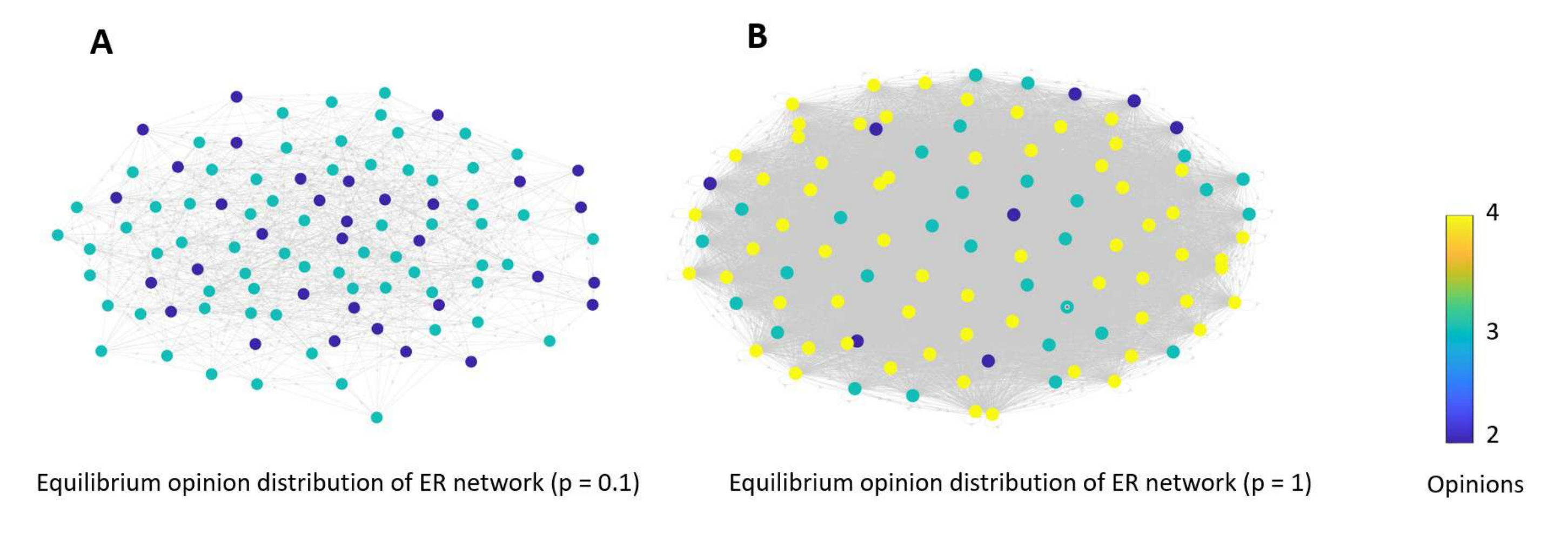}
  \caption{Network visualization of opinion distribution at equilibrium. We conduct simulations of PID opinion dynamics on 100-node ER networks, where $\theta$ is fixed at 2,  \(O = \{1, 2, \dots, 30\}\), and initial conditions are randomly generated. Light gray lines denote network edges, and node colors represent opinion values. The color-opinion mapping is provided in the color bar on the right. Panel~\textbf{A} shows a representative equilibrium at link probability $p = 0.1$, while Panel~\textbf{B} shows an equilibrium at $p = 1$.}
  \label{fig:network}
\end{figure*}

 Building upon these results, we perform further simulations: fixing \(\theta=2\), we run 1000 independent trials for \(p=0.1\) and \(p=1\), respectively, to compute the average number of nodes residing at opinion 3 at equilibrium. The results yield a mean value of 11.4  for \(p=0.1\) (while equilibria other than truth consensus feature a large fraction of nodes settled at opinion 3, such equilibria occur with an very low probability)  and 45.9 for \(p=1\). These findings are consistent with our representative case observations.

 Taken together, this offers an explanation for why the truth‑consensus probability decreases as $p$ increases when \(p>0.09\). As nodes move toward \(\theta\), large numbers of links promote the formation of local clusters near \(\theta\), which are less likely to emerge when the number of links is small. This difference becomes greater when \(\theta\) lies on the boundary of $\bO$.

We  also investigate the relationship between the optimal $p$ corresponding to the maximum truth‑consensus propability.  Specifically, we let the number of nodes $n$ take values from \(\{30,40,\dots,190,200\}\). For each $n$, we sweep over \(p\in\{0.04,0.05,\dots,1\}\). For each value of $p$, 500 independent simulation trials are performed. We record the value of $p$ that yields the maximum truth‑consensus probability with each $n$. Other conditions, such as the number of nodes and the available opinion set, are identical to those of the experiment on Erdős–Rényi networks described above.
As shown in Fig.\ref{fig:optimalp}, the optimal $p$ decreases as the number of nodes increases; for instance, the  optimal $p$ equals 0.15 for 50 nodes  and 0.06 for 200 nodes. 

\begin{figure}[htbp]
  \centering
  \includegraphics[width=\linewidth]{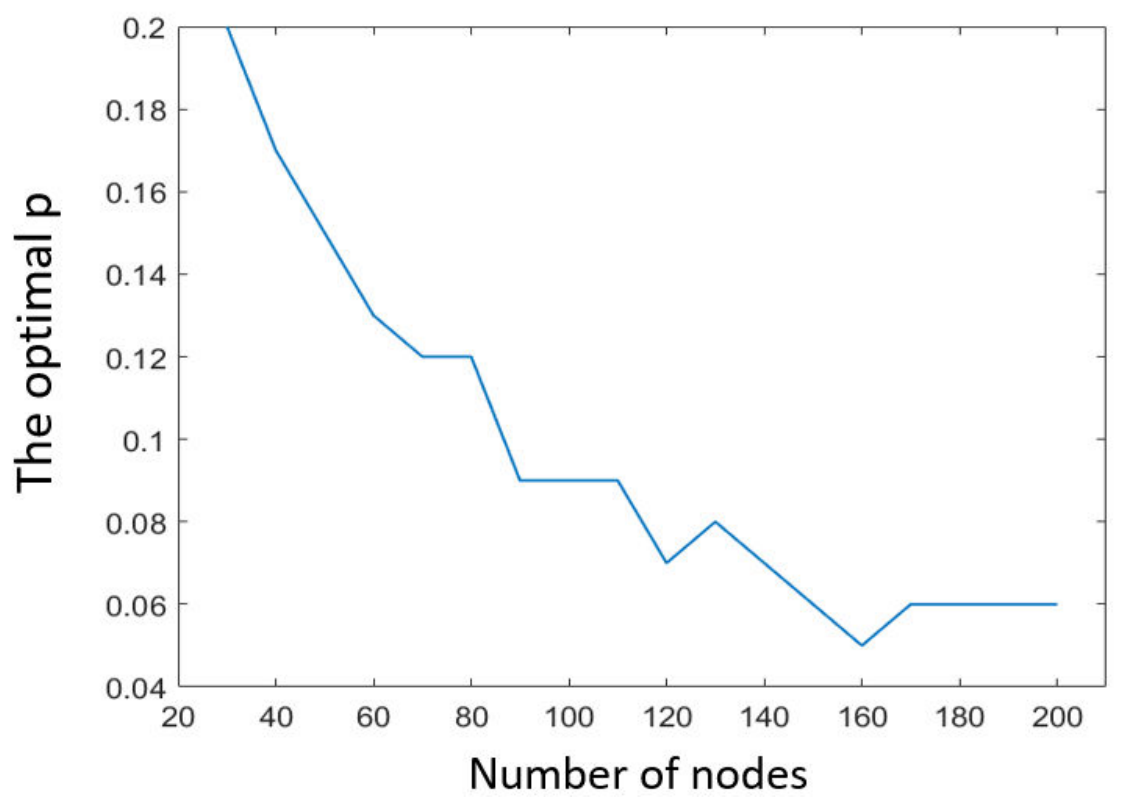}
  \caption{How the number of nodes affect the optimal $p$ corresponding to the maximum truth‑consensus rate. For each \((p,\theta)\) pair, 500 independent simulations are run to average the truth‑consensus rate and identify the $p$ yielding the maximum consensus rate.}
  \label{fig:optimalp}
\end{figure}

We further explore the relationship between the link probability $p$ and the truth consensus probability under the condition that no node initially holds $\theta$. Overall, the resulting curves shown in Fig.\ref{fig:notheta} share similar shapes with those in Fig.~\ref{fig:simulations}\textbf{B}, yet the average truth consensus probability is $0.1$ less when $p>0.1$.

\begin{figure}[htbp]
  \centering
  \includegraphics[width=\linewidth]{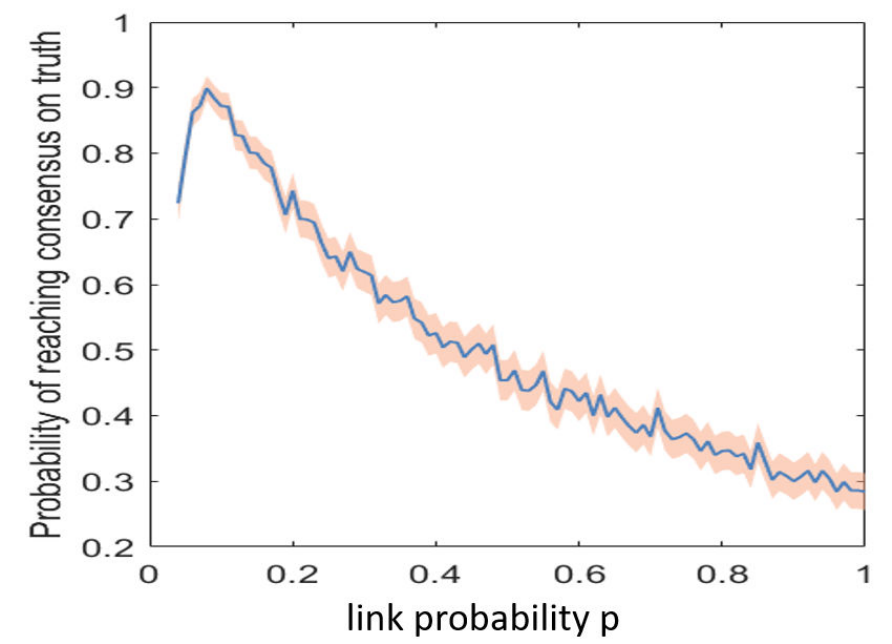}
  \caption{How link probability  affect the system’s truth‑consensus rate under the condition that no node initially holds $\theta$. The color-shaded regions demarcate the $95\%$ confidence interval for the value computed based on 1000 independent simulations under the same $p$.}
  \label{fig:notheta}
\end{figure}

Finally, we examine whether the condition in Theorem~\ref{th:cd-cs} is satisfied in simulation runs that achieve truth consensus. The results show that $100\%$ of these networks contain at least one proper strictly cohesive set. This observation suggests two implications. First, due to the random generation mechanism of E-R networks, networks satisfying the condition for almost-sure truth convergence are extremely unlikely to occur. Second, truth consensus can still be achieved with high probability even when the condition for almost-sure convergence is not satisfied.

\emph{Effects of clustering coefficient:}
To examine how clustering influences collective intelligence, we perform additional simulations on Watts–Strogatz (W–S) small-world networks~\cite{DJW-SHS:98} with 100 nodes. In W–S networks, the clustering coefficient~\cite{fagiolo2007clustering} is controlled by the rewiring probability $\beta$: smaller $\beta$ yields higher clustering. In our simulations, $\beta$ takes values from $\{0.01, 0.02, \dots, 0.99, 1\}$. For each $\beta$, 1000 independent simulations are conducted, each on a randomly generated W–S network. The initialization procedure, the edge weighting scheme and the designation of the true opinion $\theta$ are the same as in the E–R case. As shown in Fig.~\ref{fig:simulations}\textbf{E}–\textbf{H}, increasing $\beta$ (i.e., reducing clustering) leads to fewer opinion clusters at steady state and a higher likelihood of consensus on the truth. This effect is most pronounced for $\beta \in (0, 0.2)$, beyond which the steady-state behavior becomes largely insensitive to $\beta$.

To sum up, the simulation results reveal a sociologically meaningful insight: influence networks with moderate connectivity and weak clustering are most conducive to consensus on truth, while excessive clustering and over-connectivity hinder consensus on the truth.

\section{Conclusions}
This paper proposes a multi-objective optimization framework that captures the non-exchangeability between individuals’ distinct motivations during opinion updates. Unlike conventional approaches, agents in our model revise their opinions through Pareto improvement between social pressure and cognitive dissonance. We derive fundamental properties of the model: equilibrium set characterization, almost-sure finite-time convergence, and necessary and sufficient conditions for almost-sure consensus. Based on these results, we obtain graph-theoretic conditions for truth prevalence and emergence, alongside seed-selection strategies ensuring consensus on the truth. Furthermore, numerical simulations illustrate how network features including link density and clustering coefficient affect the system’s collective intelligence.

The theoretical and numerical findings yield several sociologically meaningful insights. Our model introduces multi-objective games into opinion dynamics modeling, and provides a parsimonious and mechanistic explanation for pluralistic ignorance. We further show that no network structure can guarantee consensus on truth if no individual initially expresses it. Simulations reveal that moderately sparse and weakly clustered influence networks best promote collective intelligence.

This work opens several promising research directions.  First, the single-truth setup can be generalized to scenarios with multiple truths. Second, future work may consider incomplete information, where each agent only holds partial information about the truth. Third, we can allow limited trade-offs between different costs, which relaxes the assumption of non-exchangeable motivations. These questions present exciting opportunities for future exploration.

\begin{appendices}
{

\section{Proof of Proposition~\ref{prop:traverse}}\label{app:proof-traverse}

Let $\mathcal{T} (W)$ be the directed unweighted graph such that an edge from $i$ to $j$ exists if and only if $w_{ij}>\frac{1}{2}$. We need to prove that $\calV$ is the only cohesive set in $\calV$ if and only if edges in $\mathcal{T}(W)$ form a directed cycle that traverses all the nodes in $\calV$. 

We first prove the ``if'' part. Suppose $\mathcal{T}(W)$ is a directed cycle and there exists a non-empty cohesive set $\calV_1 \subset \calV$. Let $i_1\in \calV_1$ and there is an edge from $i_1$ to some $i_2\neq i_1$ in $\mathcal{T}(W)$. By the definition of $\mathcal{T}(W)$, we have $w_{i_1i_2} >\frac{1}{2}$. Suppose $i_2 \notin \calV_1$, then $\sum_{j\in \calV_1} w_{i_1j}\leq 1-w_{i_1i_2}< \frac{1}{2}$, which contradicts that $\calV_1$ is a cohesive set. Therefore, $i_2$ must belongs to $\calV_1$. Analogously, there is an edge from $i_2$ to some $i_{3}\notin\{i_1,i_2\}$ in $\mathcal{T}(W)$ and then $i_{3} \in \calV_1$. 
Continuing in this way, we have $\{i_1,i_2,..,i_n\}\subseteq \calV_1 $. Since $\mathcal{T}(W)$ is a directed cycle traversing all the nodes in $\calV$, we have $\{i_1,i_2,...,i_n\}=\calV$, which contradicts $\calV_1 \subset \calV$.

Now we prove the ``only if'' part. Suppose $\calV$ is the only cohesive set in $\calV$. For any $i\in \calV$, divide $\calV$ into $\{i\}$ and $\calV \setminus \{i\}$. $\calV \setminus \{i\}$ is not cohesive, i.e., there exists $j\in \calV \setminus \{i\}$ such that $\sum_{k\in \calV \setminus \{i\}}w_{jk}< \frac{1}{2}$. That is, $w_{ji}>\frac{1}{2}$, which implies that there is an edge from $j$ to $i$ in $\mathcal{T}(W)$. By the arbitrariness of $i$, we have that any node in $\mathcal{T}(W)$ has at least one in-neighbor. Besides, any node $i$ in $\mathcal{T}(W)$ has at most one out-neighbor.(Otherwise, there exists $j_1, j_2 \in \calV$ that $w_{ij_1}>\frac{1}{2}$ and $w_{ij_2} > \frac{1}{2}$, which contradicts $\sum _{k\in \calV}w_{ik}=1$.) Since the total number of in-neighbors for all nodes is equal to the total number of out-neighbors in a graph, all nodes in $\mathcal{T}(W)$ have exactly one in-neighbor and one out-neighbor. This implies that $\mathcal{T}(W)$ is composed of disjoint cycles.  If $i_1,..,i_l$ form a cycle in $\mathcal{T}(W)$, by the definition of $\mathcal{T}(W)$, we have $\sum_{k=1}^l w_{i_qi_k}> \frac{1}{2}$ for any $q\in \{1,...,l\}$, and thereby $\{i_1,..,i_l\}$ is cohesive. Since  $\calV$ is the only cohesive set in $\calV$, then there is only one cycle in $\mathcal{T}(W)$. Therefore, edges in $\mathcal{T}(W)$ form a directed cycle that traverses all the nodes in $\calV$.

\section{Proof of Lemma~\ref{lemma:prop-pis}}\label{app:proof-lemma:prop-pis}
For any $z\in[\alpha,\beta]$, we have 
\begin{align*}
    |z-\theta|\le \max\{|\alpha-\theta|,|\beta-\theta|\}\le |x_i-\theta|.
\end{align*}
That is, $C^i_{cog}(z;\theta)\le C^i_{cog}(x_i;\theta)$. Let $z=\lambda \alpha+(1-\lambda)\beta$, where $\lambda\in [0,1]$.
\begin{align*}
    \Csocial^i(z;x)&=\sum_{j\in \calV}w_{ij}|x_j-(\lambda \alpha+(1-\lambda)\beta)|\\ 
    &\le \sum_{j\in \calV}w_{ij}(\lambda |x_j-\alpha|+(1-\lambda)|x_j-\beta|)\\
    &=\lambda \Csocial^i(\alpha;x+(1-\lambda)\Csocial^i(\beta;x) \\
    &\le \Csocial^i(x_i;x).
\end{align*}
The last inequality holds since $\alpha,\,\beta \in P_i(x)$. By the definition of $P_i(x)$ and the arbitrariness of $z$, we conclude that $[\alpha,\beta]\cap \bO\subseteq P_i(x)$. 

\section{Proof of Lemma~\ref{lem:ex-wt-trv}}\label{app:proof-lem:ex-wt-trv}
Denote by $S$ the update sequence along which the system reaches consensus on $0$:
    \begin{align*}    
    S=\Big(  (i_1,z_1),(i_2,z_2),...,(i_T,z_T)\Big),
    \end{align*}
where $(i_t,z_t)$ means node $i_t$'s opinion is updated to $z_t\in \bO^n$ at time $t$, and $T$ is the time step when the system reaches consensus on $0$. Let $x(t)$ be the opinion trajectory under $S$ and thereby $x_{i_t}(t)=z_t$. 

Now we construct a new update sequence $S'$ from $S$:
\begin{align*}
    S'=\Big( (i_1,z_1'),(i_2,z_2'),\dots, (i_T,z_T') \Big).
\end{align*}
That is, the node updated at each time step is the same as in $S$ but their opinions could be updated to other feasible values. Let $y(t)$ denote the trajectory generated along the update sequence $S'$, starting also from $x(0)$. The new update sequence $S'$ and the trajectory $y(t)$ are defined recursively as follows: Let $y(0) = x(0)$. For each step $k \in  \{1, \ldots, T\}$, node $i_k$'s opinion is updated to $z_k'$ given as follows
\begin{align*}
        z'_k = 
        \begin{cases}
            z_k & \text{ if }y_{i_k}(k-1) \cdot z_k > 0, \\
            0   & \text{ if }y_{i_k}(k-1) \cdot z_k \le 0,
        \end{cases}
\end{align*}
while the other nodes' opinions remain unchanged. That is, $i_k$'s opinion is updated to $z_k$ whenever such an update, along the trajectory $\{y(t)\}_{t=0}^T$, is not a crossing update. Otherwise, $i_k$ is updated to $\theta=0$. 

Denote by $t_0$ the first time step a crossing update occurs in $S$ ($1\le t_0\le T$). If $t_0$ does not exist, then the proof is already concluded since no crossing update occurs in $S$, along which the system reaches consensus on $\theta=0$. 

Suppose $t_0$ exists. Now we show by induction that $S'$ is a legal update sequence under the PID opinion dynamics. Obviously, $z_t'=z_t$ and $y(t)=x(t)$ for any $t\in \{0,\dots, t_0-1\}$. Therefore, the updates in $S'$ prior to $t_0$ (not including $t_0$) are legal. The update at time step $t_0$ is a crossing update in $S$, i.e., 
\begin{align*}
   z_{t_0}y_{i_{t_0}}(t_0-1) =  z_{t_0}x_{i_{t_0}}(t_0-1)<0
\end{align*}
Without loss of generality, assume $z_{t_0}<0<x_{i_{t_0}}(t_0-1)$. Since 
\begin{align*}
    z_{t_0} & \in P_{i_{t_0}}\big(y(t_0)\big),\text{ and}\\
    x_{i_{t_0}}(t_0-1) & =y_{i_{t_0}}(t_0-1)\in P_{i_{t_0}}\big(y(t_0)\big),
\end{align*}
by Lemma~\ref{lemma:prop-pis}, we have $0\in P_{i_{t_0}}(y(t_0-1))$ and hence $(i_{t_0},z_{t_0}')$ in $S'$ is a legal update along the trajectory $\{y(t)\}_{t=0}^T$.

Suppose $ k\ge t_0$ and $S'$ is legal for all $t\le k$. Define
\begin{align*}
   \calA_k(S)=\big{\{}i_t\,\big|\, 1\le t\le k, \,\, (i_t,z_t)\text{ is a cross update in }S\big{\}}. 
\end{align*}
By the definition of $S'$ and the induction hypothesis, we have that $\{y(t)\}_{t=1}^k$ is a feasible trajectory and
\begin{align}\label{eq:pf-lem-ex-wt-trv1}
        y_j(k)=
        \begin{cases}0, & \text{if }j\in \calA_k(S),
        \\x_j(k), & \text{if }j \in \V\setminus \calA_k(S).
        \end{cases}
\end{align}
Now we show that $(i_{k+1},z_{k+1}')$ in $S'$ is also a legal update along the trajectory $\{y(t)\}_{t=1}^{k+1}$. Let $i=i_{k+1}$. If $y_i(k)=0$, then, according to the construction of $S'$, $z_{t+1}'=0$, which must be a legal update. If $y_i(k)\neq 0$, then $y_i(k)=x_i(k)$. We split the discussion into two cases. 
  
\textit{Case 1:}  The ($k+1$)-th step in $S$ is a crossing update. Now we show $z_{k+1}'=0 \in P_i(y(k))$. Since both $x_i(k)$ and $z_{k+1}$ are in $P_i(x(k))$, and since they are of opposite signs, by Lemma~\ref{lemma:prop-pis}, we have $0\in P_i(x(k))$. As a result, we have $|z_{k+1}-0|\le |x_i(k)-0|$, and \begin{align}\label{eq:pf-lem-ex-wt-trv2}
     \sum_{j\in \calV}w_{ij}|0-x_j(k)|\le \sum_{j\in \calV}w_{ij}|x_i(k)-x_j(k)|.
\end{align}  
From $y_i(k)=x_i(k)$ and ~\eqref{eq:pf-lem-ex-wt-trv1}, we have
\begin{align}\label{eq:pf-lem-ex-wt-trv3}
\Csocial^i(y_i(k);y(k))&=\sum_{j\in \calV}w_{ij}|y_j(k)-x_i(k)|| \notag \notag \\ &=\sum_{j\in \calA_k(S)}w_{ij}|x_i(k)|\\&\,\,\,\,\,\,\,+\sum_{j\in \V\setminus\calA_k(S)}w_{ij}|x_i(k)-x_j(k)|.
\end{align}
Due to the triangle inequality,
$$\sum_{j\in \calA_k(S)}w_{ij}|x_i(k)|\ge \sum_{j\in \calA_k(S)}w_{ij}\Big(|x_i(k)-x_j(k)|-|x_j(k)|\Big).$$ Substitute this into~\eqref{eq:pf-lem-ex-wt-trv3} and according to~\eqref{eq:pf-lem-ex-wt-trv2}, we have
\begin{align*}
    \Csocial^i(y_i(k);y(k))&\ge \sum_{j\in\calV}w_{ij}|x_i(k)-x_j(k)|
    \\&\,\,\,\,\,\,\,\,-\sum_{j\in \calA_k(S)}w_{ij}|x_j(k)|\\&\ge \sum_{j\in\calV}w_{ij}|x_j(k)|-\sum_{j\in \calA_k(S)}w_{ij}|x_j(k)|\\&=C^i_{social}(0;y(k)).
\end{align*}
Combining this inequality with $\Ccog^i(0)< C^i_{cog}(y_i(k))$, we have $0\in P_i(y(k))$. This demonstrates that updating $i$ to $z_{k+1}'=0$ at time step $k+1$ is legal for $S'$.

\textit{Case 2:} The ($k+1$)-th step in $S$ is not a crossing update. Without loss of generality, assume that $x_i(k)<z_{k+1}<0$. Then we have $x_i(k)=y_i(k)<z_{k+1}'=z_{k+1}<0$. 
Since $(i,z_{k+1})$ is a legal update in $S$ along the trajectory $\{x(t)\}_{t=0}^{k+1}$, we have 
\begin{align*}
    \Ccog^i(z_{k+1}') & \le C^i_{cog}(y_i(k)),\text{ and}\\
    \Csocial^i(z_{k+1};x(k))&\le\Csocial^i(x_i(k);x(k)).
\end{align*}
The latter inequality is rewritten as
\begin{align}\label{eq:pf-lem-ex-wt-trv4}
    \sum_{j\in \calV}w_{ij}|x_j(k)-x_i(k)|\ge\sum_{j\in \calV}w_{ij}|x_j(k)-z_{k+1}|.
\end{align}
From~\eqref{eq:pf-lem-ex-wt-trv1} we have
\begin{equation}\label{eq:pf-lem-ex-wt-trv5}
\begin{aligned}
\Csocial^i(y_i(k);y(k))=&\sum_{j\in \calA_k(S)}w_{ij}|x_i(k)|\\&+\sum_{j\in \V\setminus \calA_k(S)}w_{ij}|x_i(k)-x_j(k)|.
\end{aligned}
\end{equation}
Substituting the following triangle inequality 
\begin{align*}
   |x_i(k)|-|z_{k+1}|&=|x_i(k)-z_{k+1}|\\&\ge|x_i(k)-x_j(k)|-|z_{k+1}-x_j(k)|
\end{align*}
into equation~\eqref{eq:pf-lem-ex-wt-trv5}, we have
\begin{align*}
    &\,\,\,\,\,\,\,\Csocial^i(y_i(k);y(k))\\& \ge \sum_{j\in \calA_k(S)}w_{ij}(|x_i(k)-x_j(k)|-|z_{k+1}-x_j(k)|+|z_{k+1}|)\\&\,\,\,\,\,\,\,\,+\sum_{j\in \V\setminus \calA_k(S)}w_{ij}|x_i(k)-x_j(k)|\\&=\sum_{j\in\calV}w_{ij}|x_j(k)-x_i(k)|+\sum_{j\in \calA_k(S)}w_{ij}|z_{k+1}|\\&\,\,\,\,\,\,\,\,-\sum_{j\in\calA_k(S)}w_{ij}|z_{k+1}-x_j(k)|.
\end{align*}
In turn, substituting~\eqref{eq:pf-lem-ex-wt-trv4} into the inequality above, we obtain
\begin{align*}
\Csocial^i(y_i(k);y(k))&\ge\sum_{j\in\calA_k(S)}w_{ij}|z_{k+1}|\\
&\quad+\sum_{j\in \V\setminus \calA_k(S)}w_{ij}|x_j(k)-z_{k+1}|\\&=\Csocial^i(z_{k+1}';y(k)).
\end{align*}
Combining the above inequality with $\Ccog^i(z_{k+1}')\le C^i_{cog}(y_i(k))$, we have $z_{k+1}'\in P_i(y(k))$. 

Discussions on Case 1 and Case 2 combined together demonstrate that $(i,z_{k+1}')$ in $S'$ is indeed a legal update along the trajectory $\{y(t)\}_{t=1}^{k+1}$. By induction, we conclude that 
$S'$ is a legal update sequence along the trajectory $\{y(t)\}_{t=0}^T$, with no crossing update, and equation~\eqref{eq:pf-lem-ex-wt-trv1} holds at time $t=T$. Since $x(T)=\vect{0}_n$, we have $y(T)=\vect{0}_n$, which concludes the proof.

\section{Proof of Lemma~\ref{lem:ex-pm1}}\label{app:proof-lem:ex-pm1}
The proof follows a similar approach to the proof of Lemma~\ref{lem:ex-wt-trv}: We construct a new update sequence $S'$ from a known update sequence $S$ that leads to consensus on $0$. Denote by $S=\Big( (i_1,z_1),\dots,(i_T,z_T) \Big)$ an update sequence with no crossing update, along which the system reaches consensus on $0$. The corresponding trajectory is $\{x(t)\}_{t=1}^T$. Since nodes expressing the opinion $\theta=0$ will never change their expressed opinion, we assume that, in sequence $S$, such nodes are never chosen to update their opinions. 

We now define a new update sequence $S'=\Big( (i_1,z'_1),\ldots,(i_T,z'_T) \Big)$ and the corresponding trajectory $y(t)$ starting from the initial condition $y(0) = x(0)$. For any $k \in \{ 1, \ldots, T\}$, let
\begin{align*}
    y_{i_k}(k)=z'_k = 
    \begin{cases}
        -1 & \text{if } z_k = 0 \text{ and } y_{i_k}(k-1) < 0, \\
        1  & \text{if } z_k = 0 \text{ and } y_{i_k}(k-1) > 0, \\
        z_k & \text{if } z_k \neq 0,
    \end{cases}
\end{align*}
and let $y_j(k)=y_j(k-1)$ for any $j\in \V\setminus \{i_k\}$.
    
Now we show by induction that $S'$ is legal under PID opinion dynamics. Denote by $t_0\ge1$ the first time in $S$ when a node updates its opinion to $0$. Obviously, the updates prior to $t_0$ (not including $t_0$) are identical between $S$ and $S'$, and hence legal. 

For the $t_0$-th step in $S$, without loss of generality, assume $x_{i_{t_0}}(t_0-1)\le -1$. The legality of $S$ implies $0\in P_{i_{t_0}}(x(t_0-1))$. By Lemma~\ref{lemma:prop-pis}, we have $-1\in P_{i_{t_0}}(x(t_0-1))$ and hence the update at $t_0$ in $S'$ is also legal. 

Assume $S'$ is legal for all $t\le k$, where $k\ge t_0$. Let 
\begin{align*}
    \calA_k(S) & =\{i\in \calV\,| x_i(k)\neq 0 \,\text{ or }\,x_i(0)=0\},\\
    \calB_k(S) & =\{i\in\calV\,|x_i(0)<0 \, \text{ and } \,x_i(k)=0\},\\
    \calC_k(S) & =\{i\in \calV\,|x_i(0)>0\,\text{ and }\, x_i(k)=0\}.
\end{align*}
By the definition of $S'$, we have \begin{align}\label{eq:pf-lem-ex-pm11}
        y_j(k)=\begin{cases} x_j(k), &\text{ if }j\in \calA_k(S),\\
        -1,  &\text{ if }j\in \calB_k(S),\\
        1, &\text{ if }j\in \calC_k(S).
        \end{cases}
    \end{align}

Let $i=i_{k+1}$. By symmetry, we assume $x_i(k)<0$ and hence $y_i(k)=x_i(k)<0$. Since there is no crossing update in $S$, we have $z_{k+1}\le0$. Now we show $S'$ is also legal at step $k+1$ by considering the two cases: $z_{k+1}=0$ and $z_{k+1}<0$.
   
\textit{Case 1:}  $z_{k+1}=0$ and thus $z_{k+1}'=-1$. Since $x_i(k)\in P_i(x(k))$, $x_i(k)\le -1$, and $z_{k+1}=0\in P_i(x(k))$, by Lemma~\ref{lemma:prop-pis}, we have $-1\in P_i(x(k))$, which implies
        \begin{align}\label{eq:pf-lem-ex-pm12}
            &\,\,\,\,\,\,\,\,\sum_{j\in \calA}w_{ij}|x_j(k)-x_i(k)|+\sum_{j\in \calB \cup \calC} w_{ij}|x_i(k)|\notag\\&\ge \sum _{j\in \calA}w_{ij}|x_j(k)+1|+\sum_{j\in \calB \cup \calC}w_{ij}.
        \end{align}
Substituting~\eqref{eq:pf-lem-ex-pm11} into~\eqref{eq:social}, we have
        \begin{align*}
            &\,\,\,\,\,\,\,\,\Csocial^i(y_i(k);y(k))=\sum _{j\in \calV}w_{ij}|y_i(k)-y_j(k)|\\&=\sum_{j\in \calA}w_{ij}|x_i(k)-x_j(k)|-\sum_{j\in \calB}w_{ij}(x_i(k)+1)\\&\,\,\,\,\,\,\,+\sum_{j\in\calC}w_{ij}(1-x_i(k)).
        \end{align*}
Substituting~\eqref{eq:pf-lem-ex-pm12} into the above equality, we have
        \begin{align*}
            \Csocial^i(y_i(k);y(k))&\ge \sum_{j\in \calA}w_{ij}|x_j(k)+1|+2\sum_{j\in \calC}w_{ij}\\&=\sum_{j\in \calV}w_{ij}|y_j(k)+1|\\&=\Csocial^i(z_{k+1}';y(k)).
        \end{align*}
Combining the above inequality with the fact that $\Ccog^i(z_{k+1}')\le C^i_{cog}(y_i(k))$, we have $z_{k+1}'\in P_i(y(k))$.
        
\textit{Case 2:}  $z_{k+1}=z_{k+1}'<0$. Since $S$ is legal, we have
        \begin{align}\label{eq:pf-lem-ex-pm13}
            &\quad\,\,\sum_{j\in\calA}w_{ij}|x_j(k)-x_i(k)|-\sum_{j\in \calB \cup \calC}w_{ij}x_i(k)\notag\\
            &\ge \sum_{j\in\calA}w_{ij}|x_j(k)-z_{k+1}|-\sum_{j\in \calB \cup \calC}w_{ij}z_{k+1}
        \end{align}
Substituting~\eqref{eq:pf-lem-ex-pm11} into~\eqref{eq:social}, we obtain
    \begin{align*}
        &\Csocial^i(y_i(k);y(k))\\
        &\,\,\,\,=\sum _{j\in \calV}w_{ij}|y_i(k)-y_j(k)|\\
        &\,\,\,\, =\sum_{j\in \calA}w_{ij}|x_i(k)-x_j(k)|+\sum_{j\in\calB}w_{ij}(-1-x_i(k))\\
        &\quad\,\,\,\,+\sum_{j\in \calC}w_{ij}(1-x_i(k)).
    \end{align*}
Substituting~\eqref{eq:pf-lem-ex-pm13} into the above equality, we have
    \begin{align*}
        &\Csocial^i(y_i(k);y(k))\\
        &\,\,\,\, \ge \sum_{j\in \calA}w_{ij}|x_j(k)-z_{k+1}|-\sum_{j\in \calB}w_{ij}(1+z_{k+1}) \\
        &\quad\,\,\,\,+\sum_{j\in \calC}w_{ij}(1-z_{k+1})\\
        &\,\,\,\,=\Csocial^i(z_{k+1}';y(k)).
    \end{align*}
Combining this inequality with the fact that $\Ccog^i(z_{k+1}')\le C^i_{cog}(y_i(k))$, we can conclude that $z_{k+1}'\in P_i(y(k))$.

We have shown that $S'$ is a legal update sequence without crossing updates and that~\eqref{eq:pf-lem-ex-pm11} holds for $k=T$. That is, every $i\in \calV_{<0}(x(0))$ expresses $-1$ and every $i\in \calV_{>0}(x(0))$ expresses $1$ at time $T$ along $S'$. This concludes the proof.  

\section{Proof of Lemma~\ref{lem:equi-pm10}}\label{app:proof-lem:equi-pm10}

We split the proof into two cases.

\emph{Case 1:} there exists a non-empty strictly cohesive set $\calA \subseteq \calV_{=-1}(x(0))$. We now prove that, regardless of the update sequence, nodes in $\calA$ will never alter their states. Otherwise, let $i\in \calA$ be the first one to alter its state, and this update occurs at some time $t+1$. However, by calculating $\Csocial^i(z;x(t))$, we could find that 
\begin{align*}
    &\,\,\,\,\,\,\,\,\Csocial^i(-1;x(t))-\Csocial^i(0;x(t))\\&=\sum_{j\in \calV_{=0}(x(t))}w_{ij}+\sum_{j\in \calV_{=1}(x(t))}w_{ij}-\sum_{j\in \calV_{=-1}(x(t))}w_{ij}\\&=1-2\sum_{j\in \calV_{=-1}(x(t))}w_{ij}\le1-2\sum_{j\in \calA}w_{ij}<0
\end{align*}
and
\begin{align*}
    &\,\,\,\,\,\,\,\Csocial^i(-1;x(t))-\Csocial^i(1;x(t))\\&=2\sum_{j\in \calV_{=1}(x(t))}w_{ij}-2\sum_{j\in \calV_{=-1}(x(t))}w_{ij}\\&\le2-4\sum_{j\in \calA}w_{ij}<0.
\end{align*}
This implies that neither $0$ nor $1$ is in $P_i(x(t))$, i.e., $i$ cannot update its opinion to $0$ or $1$, which is a contradiction. Therefore, in this case, all nodes in $\calA$ stay in $-1$ in all equilibria of the system.    

\emph{Case 2:} $\calV_{=-1}(x(0))$ does not contain any non-empty strictly cohesive set. We now construct an update sequence along which the system achieves consensus at 1. Consider the cohesive expansion of $\calV_{=1}(x(0))$. Since $\calV_{=-1}(x(0))\cup \calV_{=1}(x(0))=\V$, if $\E(\calV_{=1}(x(0)))\neq\calV$, then there exists a non-empty node set $\calA\subseteq\calV_{=-1}(x(0))$ such that, for any $i\in \calA$, $ \sum_{j\in\calA }w_{ij}>\frac{1}{2}$, which contradicts the pre-assumption that $\calV_{=-1}(x(0))$ contains no non-empty strictly cohesive set. Therefore, we have $\E\big(\calV_{=1}(x(0))\big)=\calV$. Let $(i_1,\dots, i_T)$ be a node-addition sequence for the cohesive expansion $\E\big(\calV_{=1}(x(0))\big)$. Based on the Pareto-improvement mechanism, one could easily check that, if $\calV_{=0}(x(0))$ is empty, such a node addition sequence in fact induces a legal opinion-update sequence 
\begin{align*}
    S=\Big(  (i_1,1),\,(i_2,1)\dots,(i_T,1)\Big)
\end{align*}
of the PID opinion dynamics. Moreover, $\E\big(\calV_{=1}(x(0))\big)\!=\!\calV$ implies that $x_i(T)=1$ for any $i\in \calV$ when the opinion-update sequence $S$ terminates at time step $T$. That is, there exists a finite-time update sequence driving the system to a false consensus. This together with the arguments in Case 1 concludes the proof.


}

\section{Proof of Corollary~\ref{cor:grn} \label{app:proof-cor:grn}}
By Theorem~\ref{th:cs-at-th}, it suffices to show that this globally reachable node (denoted by $i$) lies in every strictly cohesive set. For a certain strictly cohesive set $\calA$, let $j \in \calA$. By the given topological condition, there exists a directed path from $j$ to $i$ consisting exclusively of edges with weights greater than $0.5$. Denote this path as $j \to j_1 \to j_2 \to \dots \to i$. Since $w_{jj_1} > 0.5$, it follows from the definition of  strictly cohesive set that $j_1 \in \calA$. By the same reasoning, $j_2 \in \calA$. Iterating this argument along the path, we conclude that $i \in \calA$. 
\end{appendices}

\bibliographystyle{IEEEtran} 
\bibliography{ref}

\end{document}